\documentclass[acmtog,nonacm]{acmart}

\usepackage{epigraph}
\usepackage{nth}
\usepackage{algorithm}
\usepackage[noend]{algpseudocode}

\usepackage{booktabs}
\usepackage{graphicx}
\usepackage{soul}
\usepackage{microtype}
\usepackage{xspace}
\usepackage{natbib}
\usepackage{multirow}
\usepackage{url}
\usepackage[nolist]{acronym}
\usepackage{algorithm}
\usepackage{algpseudocode}
\usepackage{xcolor}
\usepackage{siunitx}
\usepackage{IEEEtrantools} 
\usepackage[makeroom]{cancel}

\newtheoremstyle{mytheoremstyle}%
  {}              %
  {}              %
  {}                        %
  {0pt}                  %
  {\bfseries\sffamily}   %
  {.}                    %
  {5pt plus 1pt minus 1pt} %
  {\thmname{#1}\thmnumber{ #2}%
   \thmnote{ \normalfont\rmfamily(#3)}}                     %
\theoremstyle{mytheoremstyle}
\newtheorem{theorem}{Theorem}

\newtheorem*{remark*}{Remark}   %

\makeatletter
\renewenvironment{proof}[1][\proofname]{\par
  \pushQED{\qed}%
  \normalfont \topsep6\p@\@plus6\p@\relax
  \trivlist
  \item[\hskip\labelsep
        \bfseries\sffamily          %
    #1\@addpunct{.}]\ignorespaces
}{%
  \popQED\endtrivlist\@endpefalse
}
\makeatother

\definecolor{classicColor}{HTML}{000000}
\definecolor{correctedColor}{HTML}{FF7744}

\usepackage{pifont}

\newcommand{\change}[1]{#1}

\DeclareGraphicsExtensions{.png,.jpg,.pdf,.ai,.psd}
\def\figurePath{figures/}

\keywords{noise generation, Monte Carlo method, video diffusion models, non-photorealistic rendering, texture synthesis}

\newcommand{\eg}{e.g.,\xspace}
\newcommand{\ie}{i.e.,\xspace}

\newcommand{\refFig}[1]{Fig.~\ref{fig:#1}}
\newcommand{\refTab}[1]{Tab.~\ref{tab:#1}}
\newcommand{\refSec}[1]{Sec.~\ref{sec:#1}}
\newcommand{\refEq}[1]{Eq.~\ref{eq:#1}}
\newcommand{\refAlg}[1]{Alg.~\ref{alg:#1}}
\newcommand{\refThm}[1]{Thm.~\ref{thm:#1}}

\usepackage{icomma}

\newcommand{\myfigure}[2]{%
    \begin{figure}[htb]%
    \centering\includegraphics*[width = \linewidth]{\figurePath#1}%
    \caption{#2}%
    \label{fig:#1}%
    \end{figure}%
}

\newcommand{\mycfigure}[2]{%
    \begin{figure*}[htb]%
    \centering\includegraphics*[width = \linewidth]{\figurePath#1}%
    \caption{#2}%
    \label{fig:#1}%
    \end{figure*}%
}

\soulregister\ref7
\soulregister\cite7
\soulregister\refFig7
\soulregister\refAlg7
\soulregister\cite7
\soulregister\ref7
\soulregister\pageref7
\soulregister\shortcite7
\soulregister\eg0
\soulregister\ie0
\soulregister\etal0

\DeclareGraphicsExtensions{.png,.jpg,.pdf,.ai,.psd}
\newcommand{\mysection}[2]{\section{#1}\label{sec:#2}}
\newcommand{\mysubsection}[2]{\subsection{#1}\label{sec:#2}}

\definecolor{colorA}{HTML}{4285f4}
\definecolor{colorB}{HTML}{ea4335}
\definecolor{colorC}{HTML}{fbbc04}
\definecolor{colorD}{HTML}{34a853}
\definecolor{colorE}{HTML}{ff6d01}
\definecolor{colorF}{HTML}{46bdc6}
\definecolor{colorG}{HTML}{000000}
\definecolor{colorH}{HTML}{777777}
\definecolor{colorI}{HTML}{bdd6ff}
\definecolor{colorJ}{HTML}{6a9e6f}

\begin{acronym}
\acro{BRDF}{bi-directional reflectance distribution function}
\acro{RDS}{random-dot stereogram}
\acro{MC}{Monte Carlo}
\acro{QMC}{Quasi-Monte Carlo}
\acro{NPR}{non-photorealistic rendering}
\acro{MSE}{mean squared error}
\end{acronym}

\newcommand{\refSecSuppl}[2]{Suppl.~Sec.~#2}

\usepackage{bm}

\newcommand{\mymath}[2]{
    \newcommand{#1}{\TextOrMath{$#2$\xspace}{#2}}
}

\begin{document}

\title{Tabula Rasa: Monte Carlo estimation of unit-variance noise with controlled spatio-temporal correlation}

\author{Tobias Ritschel}
\affiliation{%
	\institution{University College London}
	\country{UK}
}
\affiliation{%
	\institution{Meta Reality Labs}
	\country{USA}
}
\email{t.ritschel@ucl.ac.uk}

\author{Yang Zhou}
\affiliation{%
	\institution{Meta Reality Labs}
	\country{USA}
}
\email{zhouyang@meta.com}

\author{Nick Milef}
\affiliation{%
	\institution{Meta Reality Labs}
	\country{USA}
}
\email{nickmilef@meta.com}

\author{Mikhail Dereviannykh}
\affiliation{%
	\institution{Karlsruhe Institute of Technology}
	\country{Germany}
}
\affiliation{%
	\institution{Meta Reality Labs}
	\country{USA}
}
\email{devmishok@gmail.com}

\author{Chen Liu}
\affiliation{%
	\institution{University College London}
	\country{UK}
}
\email{chen.liu.21@ucl.ac.uk}

\author{Christophe Hery}
\affiliation{%
	\institution{Meta Reality Labs}
	\country{USA}
}
\email{christ.hery@gmail.com}

\author{Carl Marshall}
\affiliation{%
	\institution{Meta Reality Labs}
	\country{USA}
}
\email{csmarshall@meta.com}

\begin{abstract}
We suggest a method to generate time-varying Gaussian noise with controlled variance and controlled temporal correlation.
This noise is used in several downstream tasks for temporal control and temporal coherence.
The core technical idea is to phrase this problem as joint Monte-Carlo estimation of both a classic pixel reconstruction and estimation of variance using the concept of ``sketching'' from the database literature.
We demonstrate that our method allows temporal control for downstream tasks with simpler and faster code than previous methods.
\end{abstract}

\maketitle

\mymath{\pixelValue}{P}
\mymath{\correctedPixelValue}{\pixelValue_\mathrm{Cor}}
\mymath{\otherImageCoord}{\mathbf y}
\mymath{\pixelDomain}{\Omega}
\mymath{\pixelSubdomains}{\pixelDomain_i}
\mymath{\seed}{\xi}
\mymath{\var}{\mathrm{Var}}
\mymath{\world}{f}
\mymath{\cardinality}n
\mymath{\sampleCount}m
\mymath{\subIndex}i
\mymath{\sampleIndex}j
\mymath{\area}{a_\subIndex}
\mymath{\worldConstant}{F_\subIndex}
\mymath{\estimate}{\widehat}
\mymath{\sample}{\omega_\sampleIndex}
\mymath{\hashing}{\mathcal{H}}
\mymath{\bigHashing}{\hashing_\mathrm b}
\mymath{\smallHashing}{\hashing_\mathrm s}
\mymath{\lattice}{\mathcal L}
\mymath{\noise}{\mathcal N}
\mymath{\bin}b
\mymath{\binCount}o
\mymath{\prob}p

\section{Introduction}
We propose a method\footnote{\url{https://github.com/facebookresearch/Tabula-Rasa}} to generate image sequences of Gaussian noise that have unit variance and spatio-temporal correlation induced by light transport.
If, for example, one pixel is reflecting the same world space position which is visible in another pixel, the two pixels have to be the same.
If a world space position is visible in one frame, occluded in the next one and becomes visible in the third frame again in a different position, the value in the first and last frame has to be the same for the pixel it falls on.
Enforcing such correlation, however, will affect its variance, and the challenge is to combine the aims of correlation and variance control.

Time-correlated unit-variance noise has an important use case in diffusion-based image video synthesis and reconstruction 
\cite{burgert2025gowiththeflow,chang2024how}, but also can be used for video stylization, or animated texture synthesis.

\begin{figure}[htbp]
    \includegraphics[width=\columnwidth]{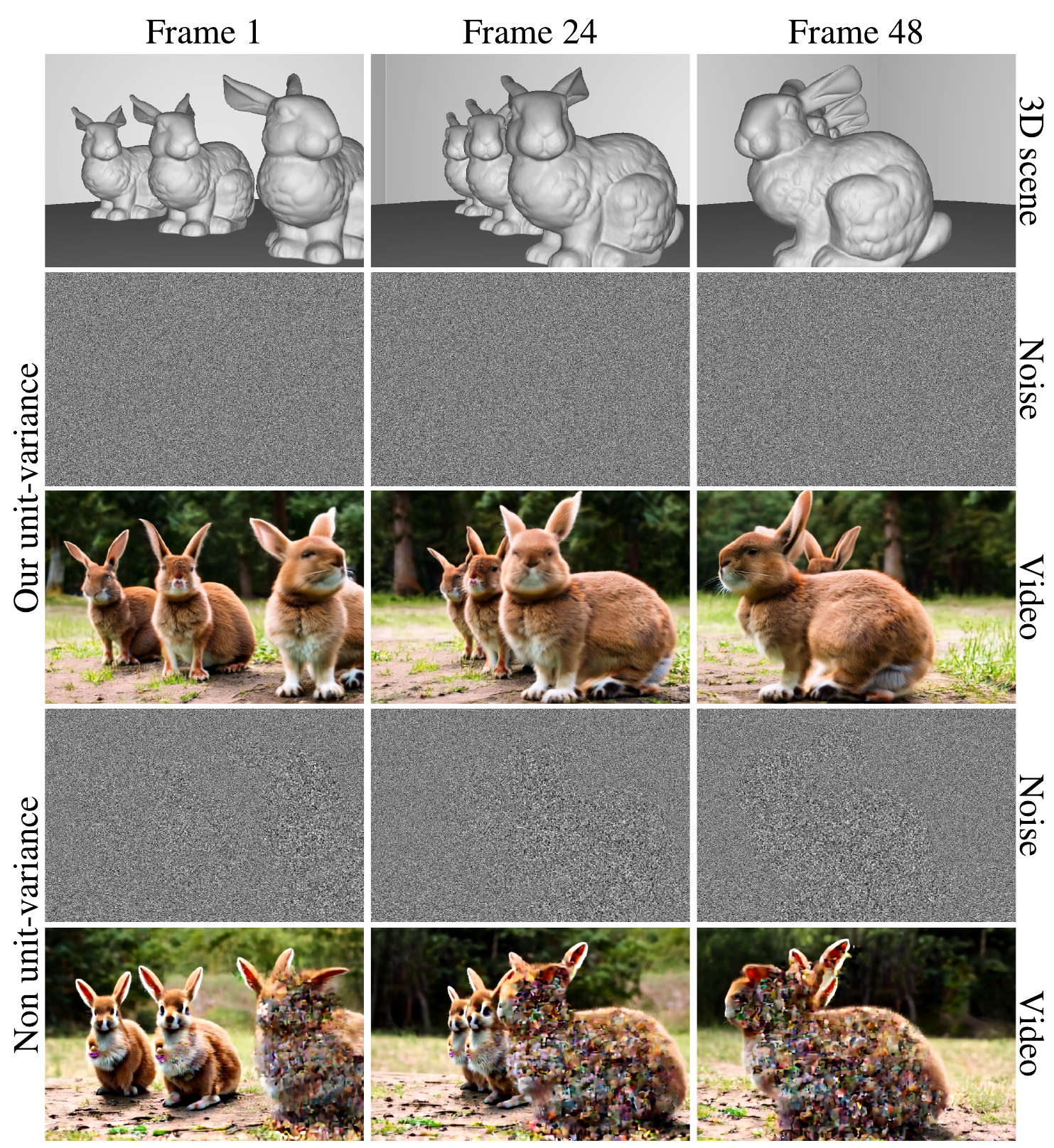}
    \caption{A sample video generated by our method driven by geometry of bunnies.
    Maintaining temporal correlation of noise patterns is critical for controlling video generation, but maintaining unit variance of the noise is necessary to produce high-quality video. 
    Generating a $1024^2$ frame takes \qty{2.44}{\milli\second} on an Nvidia A100 and adds three lines of code to an off-the-shelf Monte Carlo path tracer.}
    \label{fig:Teaser}
\end{figure}

Previous attempts to produce such noise rely on combinatorial enumeration of all variates by means of rasterization \cite{chang2024how}, Brownian bridges \cite{deng2025infiniteresolution} or maintaining a minification or magnification state \cite{burgert2025gowiththeflow}.
These have complex implementations, can be slow to compute as they operate sequentially: to produce frame $t+1$ one needs to know frame $t$.
They also rely on 2D optical flow, which will struggle with clipping and occlusion, while we produce noise that is coherent ``in 3D'' by construction.

Our solution is based on \ac{MC}-based pixel reconstruction, where the integration is in the image domain \cite{smith1995pixel} and the integrand is noise \cite[Sec. 7.8]{pharr2023physically}.
Doing so right away, unfortunately, would not produce unit variance.
This article is about augmenting \ac{MC}-based pixel reconstruction with means to preserve said variance.
The main idea relies on means of counting contribution of variates to an estimate, inspired by database concepts, such as linear counting \cite{whang1990linear} and hyper-log-log \cite{flajolet1985probabilistic}.

\section{Previous work}
Patterns of spatially correlated noise were studied by \citet{julesz1971foundations} since the 1950s.
A particular example are \acp{RDS}, where the correlation over space is such as it would be caused by binocular disparity and hence eliciting a sensation of depth under appropriate ocular vergence.
Where the naked eye sees only noise, under a different condition (ocular vergence), a strikingly strong perception emerges.
Our method is the time analogy of an \ac{RDS}, where a static image does not reveal any structure, in the right condition (motion) the pattern is apparent to man and machine.

A second inspiration is the question what a pixel really is \cite{smith1995pixel} and a typical answer is to consider the integration of a reconstruction filter and a signal \cite{pharr2023physically}.
Assuming a reconstruction filter fixed, 
\eg a 2D box, we would ask what the 3D light field signal needs to be such that we can efficiently sample from it and it retains the desired variance.

In general, preservation of image statistics is used in advection visualization \cite{neyret2003advected}, to suppress blending artifacts in texture synthesis \cite{heitz2018high}, and in recent pyramidal image generation \cite{wang2024generative}.
A final use case is \ac{NPR}  where methods by \citet{benard2011state} or \citet{kass2011coherent} retain the 3D temporal coherence, but at the expense of  statistics that change under magnification or minification.

\citet{ge2023preserve} found that the noise produced by inverting an image is not the same as the noise found when inverting a video but correlated over time instead.
Hence, feeding correlated noise into a synthesis should encourage correlated results \cite{ceylan2023pix2video}.
The question then is how to produce such noise?

A principled answer is given by \citet{chang2024how}, who derive a sophisticated image warp.
To produce correlated frames,  noise in the first frame is moved along arbitrary flow by rasterizing triangles to subpixels.
This preserves unit variance, but at the cost of substantial code complexity and compute time.

\change{Another} attempt to produce noise by warping is due to \citet{burgert2025gowiththeflow}.
They expand and collapse random pixels in a controlled fashion.
When zooming in, a random variable is split in two, when zooming out, two variables are combined into one, all with the correct handling of variance.
This approach is limited to 3D flows, and while its per-frame cost is similar to ours, it is sequential, \ie frame $t$ needs to be computed before frame $t+1$.

In summary, our method achieves state-of-the-art speed while leveraging existing theory of \ac{MC} image synthesis, leading to more 3D-coherent noise and simpler code than approaches using 2D flows. Importantly, we support random access, allowing parallelization over time (on four A100s, our method is simply four times faster).

\section{Our approach}
In our notion, a single pixel at a single video frame is a random variable that can be reconstructed as 
\begin{equation}
\pixelValue(\seed) = 
\int_\pixelDomain
\world(\otherImageCoord)(\seed)
\mathrm d\otherImageCoord
,
\end{equation}
where
\seed is the random seed and \otherImageCoord is an integration variable over the pixel domain \pixelDomain.
Changing the seed will produce new realizations of a random world variable \world, the projection of an arbitrary complex 3D noise structure to 2D (\refFig{Footprint}).

This is identical to Monte Carlo reconstruction path-traced pixels \cite{pharr2023physically}, just that the integral is a random variable.
It can be imagined as ray-tracing different procedural worlds with different random seed \seed and using the resulting statistics.

\subsection{Monte Carlo solution without controlled variance}
A realization of this random variable can be estimated by fixing \seed and solving the integral using Monte Carlo techniques:
\begin{equation}
\label{eq:nonUnitEstimator}
\estimate\pixelValue(\seed) = 
\sum_{j=1}^\sampleCount
\frac
{\world(\sample)(\seed)}
{\prob(\sample)}
,
\end{equation}
where \sample is a sample of the pixel footprint, \prob the probability of that sample and \sampleCount is the number of samples.
In the limit, this value converges to the true random variable with more samples.
All sophistication of Monte Carlo rendering can be applied here \eg using \ac{QMC} sequences such as Hammersley.

While the estimator in \refEq{nonUnitEstimator} is unbiased, the variance of the random variable is not 1, because multiple variables that are unit-variance are combined with different proportions.
We could remedy this situation if we knew the variance, simply by re-scaling
\begin{equation}
\estimate\correctedPixelValue
= \frac{\estimate\pixelValue}{\sqrt{\estimate{\var(\pixelValue)}}}
.
\end{equation}
Hence, the problem becomes the problem of estimating the variance.
Note, that this is not about the variance of the estimator $\var(\estimate\pixelValue)$, but estimation of the variance $\estimate{\var(\pixelValue)}$ (ideally, also with low variance of the variance estimator, $\var(\estimate{\var(\pixelValue)})$).

We provide more background to the notion of variance and estimation in \refSecSuppl{Primer}{1}.

\subsection{Monte Carlo solution with controlled variance}
To control variance, we will make two assumptions: First, that there are only finitely many random variables inside the integration domain.
This means, each random variable is constant in respect to spatial coordinate over a part of the pixel domain \pixelDomain, a piecewise-constant random field.
Second, that we can classify each discrete 2D pixels sample position to one of these discrete random variables.

\myfigure{Footprint}{
Footprint of a single output pixel (blue) in the random field for 2D (left) and 3D (right).
For 3D, two surfaces fall onto one pixel, resulting in a discontinuous (pink) 3D correlation a 2D flow cannot express.
}

\subsubsection{Sub-pixel decomposition}
The first assumption  is visualized in \refFig{Footprint} that shows how the reconstructed filter value is a sum of \cardinality discrete sub-integrals
\begin{equation}
\pixelValue(\seed)
= 
\sum_{\subIndex=1}^\cardinality
\int_{\pixelDomain_\subIndex}
\world_\subIndex(\otherImageCoord)(\seed)
\mathrm d
\otherImageCoord
.
\end{equation}

As the function is constant within a subdomain, this means the result is the product of that constant function value \worldConstant and the area of the subdomain \area:
\begin{equation}
\pixelValue(\seed)
= 
\sum_{\subIndex=1}^\cardinality
\area
\worldConstant(\seed)
.
\end{equation}

We do not assume anything about the cardinality \cardinality, and do also not need to know it for our approach to work.
We also do not know the shape of the subdomains $\pixelSubdomains$.
It would depend on \world.
In particular, if \world operates in 3D, it might involve effects like occlusion and perspective or even reflection and refraction,  which will result in a complex partition of the pixel domain \pixelDomain into arbitrarily shaped subdomains \pixelSubdomains.

\subsubsection{Variance}
At any rate, three properties hold:
First, the values of \worldConstant are uncorrelated.
This is just a property of the underlying random field.
Second, $\sum\area=1$, so the sum of all parts of a pixel is the one of the total pixel, without loss of generality, normalized to 1.
Third, $\area\geq0$, area is positive.
Hence, the estimate is a convex combination of independent random variables \worldConstant with weights \area and the variance of the combined random variable is
\begin{equation} \label{eq:varianceEstimator}
\var(\pixelValue) 
=
\sum_{\subIndex=1}^\cardinality
\area^2
\underbrace{
\var(\worldConstant)
}
_{=1}
=
\sum_{\subIndex=1}^\cardinality
\area^2
.
\end{equation}

So we are left with the problem of estimating the area \area of all \cardinality subdomains, enabling a unit-variance estimator
\begin{equation}
\label{eq:unitVarianceEstimator}
\estimate\correctedPixelValue = 
\frac{\estimate\pixelValue}{\sqrt{\sum_{\subIndex=1}^\cardinality \estimate{\area}^2}}.
\end{equation}
The key innovation of this article is to perform this estimation efficiently in practice.

\subsubsection{Estimating variance}
If we only had a low and known number of known subdomains, one could simply estimate both the value and the area using MC in isolation, apply the variance correction equation \`a la  $(\sum a_i^2)^{-1/2}$ and move on.
Unfortunately, we cannot know the subdomain shapes; we do not even know their number.

The key to make estimation of variance still work in this condition is based on \emph{sketching}, a term from the database literature \cite{flajolet1985probabilistic,whang1990linear} to describe a compact representation of a more complex statistic.
More concretely, our sketch is a combination of two levels of \emph{hashing} with a \emph{histogram}.

We estimate the corrected value $\estimate\correctedPixelValue$ using an extended form of Monte Carlo, which jointly computes both the integral $\estimate\pixelValue$ for a fixed seed \seed and the areas $\estimate{\area}$ which are then combined according to \refEq{unitVarianceEstimator}.

First, we pick a sample \sample and compute the integrand for that sample, $\world(\sample)$.
We represent the integrand in the form of \begin{equation}
\world(\seed)(\sample)
= 
\noise(\bigHashing(\lattice(\sample)) \oplus \seed),
\end{equation}
where 
\noise, 
\bigHashing and
\lattice are the \emph{noise}, \emph{hashing}, and \emph{lattice} functions, respectively \change{and $\oplus$ is an operator to combine hash seeds}.
This triplet of functions is just what a by-the-book \cite{ebert2002texturing,lagae2006long} procedural noise shader would do, but we need to dissect it to query the right information as explained next:

\paragraph{Lattice}
The lattice function \lattice  uses a high-resolution ($1024^3$) lattice, to convert the continuous coordinate to a discrete one.
This lattice does not need to be limited to 2D.
For correlation induced by ray-tracing, the lattice function would return the 3D hit position in a 3D lattice.
We can support any form of ray-tracing, as long as it returns one unique hit point inside this lattice: primary rays, perfect reflection or refraction, but no distribution effects or bounces.
The lattice is infinite and never represented in memory in any form.

\paragraph{Hashing}
The hash function \bigHashing maps the 2D or 3D lattice coordinate to a single scalar integer code of 32-bit precision.
We use the \texttt{XOR} between the product of each dimension and a large but representable prime number as popularized by \citet{teschner2003optimized}.

\paragraph{Noise}
Finally, the noise function \noise uses the hashed lattice code as a key to sample a unit Gaussian.
Using the same key results in the same value.
The key given to noise is the suitable combination (depicted with $\oplus$, an operator to combine random keys) of the lattice key, and a global key.
The lattice key makes the random value change over space (random field) while the global key makes this change depending on the seed.
For a fixed global key, areas with the same lattice value are constant and these areas form the subdomains \pixelSubdomains.

Coming back to the normal Monte Carlo part of our estimator, \seed is fixed and so far we simply sample a continuous coordinate \sample, run \lattice, \hashing and \noise on it and record the random variate.
We sum up those variates and divide by the sample count.
This produces samples from a random variable with desired correlation but not with unit variance.
To tackle this, we also compute variance in respect to \seed in the same Monte Carlo loop.

If we were able to represent a sufficiently large histogram of size $\cardinality=2^{32}$ to distinguish all values of $\bigHashing(\lattice(\sample))$, we could just keep track of the number of samples that fell into each of the \cardinality random variables.
The number of samples per variable, divided by the total number of samples, would then be an estimator of each area $\estimate\area$.
Obviously, it is infeasible to use storage in the order of $2^{32}$ per pixel.

\mymath{\histogram}w

To reduce the memory requirement for storing the histogram, we apply a different (``small'') hash function \smallHashing to compute a bin index $
\bin = \smallHashing(\lattice(\sample))$.
This hashing reduces the many discrete lattice coordinates to very few values, \eg $\binCount=256$.
Few enough to maintain a histogram of all of them in a per-GPU thread buffer.

During Monte Carlo estimation, for every sample, we compute \bin and increment the histogram accordingly.
After estimation has completed, we divide all histogram values by \sampleCount, making the histogram a categorical distribution (positive values that sum to one in no particular order). 
We now know (modulo hash collisions) the weight of all \binCount different variates that are combined.
With this information we can compute the resulting variance as $\sum a_i^2$.

For \smallHashing, we simply, but non-obviously, use $\smallHashing(x,y,z)=x + 16y$, \ie a repeating pattern of $16\times 16$ blocks that entirely ignores depth, as different lattice $z$ values for the same lattice $x,y$ are resolved by occlusion.
\bigHashing needs to respect depth, so as to support perspective, parallax, occlusion and all that.
Certainly, \smallHashing can lead to collisions, and these would lead to double-counting.
Fortunately, these only happen at extreme minification or magnification around 16-to-1.

Why do we employ two hashings and not just a single one?
Recall that we need to define the random field, alongside a sketch for the variance: we thus have two quite different purposes.
If we used the resolution of the random field (millions) to estimate variance, the histogram would use far too much memory.
If we only used the resolution of the histogram (100s), the values would repeat, a particularly unwanted spatial correlation.

\mysubsection{Prefiltering}{Prefiltering}

\paragraph{Problem}
While the above works, it has a sampling issue.
Consider a single screen pixel that would be subject to minification, that is, many random variables map to a single pixel, \eg to 100 random variables.
The estimators presented so far do handle this by ``counting'' the 100 variables to re-scale the variance to the desired level. This is an estimator of the variance, but the variance of the estimator itself will simply be too high:
Obviously, counting 100 variables is unlikely to be accurate with 10 samples. In general, the weighted average of more variables needs more correction due to the law of large numbers.
\paragraph{Solution intuition}
A better solution would be to only have fewer values to count, say 10, which will have a lower estimator variance and less correction.

This is similar to the general problem of prefiltering for procedural textures.
Our method is just a sophisticated procedural texture, but ultimately is subject to filtering considerations, too.
So instead of averaging a high number of high-frequency procedural texels, one can also average a low number of prefiltered texels.

\paragraph{Clamping}
A classic method from the graphics literature has shown to solve our issue: clamping \cite{norton1982clamping}.
The idea is that if a procedural texture is a combination of different levels of detail (``bands''), one can simply clamp (set to zero, if their mean is zero anyway) the bands that will be averaged away.

Hence, first, we determine the pixel magnification factor.
It is larger than one, if more than one noise pixels falls onto a pixel, or less if fewer pixels do.
It can be computed as the cross product of the columns of the Jacobian of the mapping.
A mapping to be used with this variant of our mapping is supposed to provide this interface.
When using our method with ray-tracing implemented in a framework supporting differentiation (such as JAX), this is supported.
In graphics APIs that support rasterization, one can simply call the fast operations \texttt{dFdx} and \texttt{dFdy} \cite{OpenGLSpec} and use the Jacobian determinant of these differentials.
Knowing the pixel footprint, we can use the approach of \citet{norton1982clamping} to separate the noise levels that indeed matter, and those that will just get averaged away.
The relevant noise level is selected using the logarithm of the footprint, similar to MIP level selection in OpenGL.

\paragraph{Tri-linearity}
Instead of using a single most-relevant noise resolution, we apply the analogy of tri-linear MIP mapping: we compute the two most relevant levels and blend them, ignoring all others.
As both levels hold independent random variables, this has to be accounted for when computing variance again.
So if, for example, one random variable from level $l$ has the trilinear weight of 0.2, and four random variables from level $l+1$ fall onto one pixel with weight 0.8/4, each, the expected variance is $0.2 + 4\times 0.2^2=0.36$.

\mymath{\sampling}{\mathcal S}
\mymath{\latticeCoord}{\mathbf z}

\mysubsection{Summary}{Summary}
The resulting method is summarized in \refAlg{Main}.
A variant using scatter-type rasterization computation and prefiltering is provided in \refSecSuppl{Pseudocode}{2}.

\begin{algorithm}[htb]
    \caption{Path tracing variant.}
    \label{alg:Main}
    \begin{algorithmic}[1]
    \Procedure{Estimate}{\pixelDomain, \seed}
        \State \pixelValue = 0
        \color{correctedColor}
        \State \histogram = \Call{CreateHistogram}{\binCount}
        \color{classicColor}
        \For{\sampleIndex = 0\ldots\sampleCount}
            \State $\sample$ = \Call{Sample}{\pixelDomain, \sampleIndex}
            \State $\latticeCoord$ = $\lattice(\sample)$
            \State $\pixelValue$ = $\pixelValue + \noise(\bigHashing(\latticeCoord) \oplus \seed)$
            \color{correctedColor}
            \State $\histogram$ = \Call{IncrementBin}{$\histogram, \smallHashing(\latticeCoord)$}
            \color{classicColor}
        \EndFor
        
        \State \Return{\change{$(\pixelValue/\sampleCount) \cdot (\sum(\histogram/\sampleCount)^2)^{-1/2}$}}
    \EndProcedure
    \end{algorithmic}
\end{algorithm}

\renewcommand*\Call[2]{\textproc{#1}(#2)}

\mysection{Analysis}{Analysis}
We study the convergence both analytically (\refSec{Analytical}) and empirically (\refSec{Empirical}), and measure computational performance (\refSec{Performance}). 

\mysubsection{Analytical}{Analytical}
We show by mathematical deduction that our estimator is consistent (\refThm{consistency}), and bound its bias (\refThm{bias}).

\mymath{\mapping}{\varphi}
\mymath{\mappingNom}{u}
\mymath{\mappingDenom}{v}

\begin{theorem}[Consistency] \label{thm:consistency}
$\estimate\correctedPixelValue=\correctedPixelValue$
once
$\sampleCount\to\infty$.
\end{theorem}
\begin{proof}
The Continuous Mapping Theorem~\citep[Therorem~2.3]{vanderVaart1998} states that a continuous mapping \mapping preserve limits even if its arguments are random variables.
\refEq{unitVarianceEstimator} is exactly such a continuous mapping $\mapping(\mappingNom, \mappingDenom) = \mappingNom / \sqrt{\mappingDenom}$, which takes two random variables $\mappingNom=\estimate\pixelValue$
and 
$\mappingDenom=\estimate{\var(\pixelValue)}$ as arguments.
The limits of $\mappingNom$ and $\mappingDenom$ are \pixelValue and $\var(\pixelValue)$, respectively: \mappingNom is the common path tracing pixel estimator, known to be consistent.
The other, \mappingDenom, is integration of a piecewise constant function, that trivially is consistent as long as there is no hash collisions due to \smallHashing.
Finally, $\mapping$ is smooth, considering $\mappingDenom>0$ since the number of sub-pixel domains is finite.  
\end{proof}

\mymath{\magic}{\psi}
\mymath{\areaEstimates}{\estimate{\mathbf a}}
\mymath{\trueAreas}{\mathbf a}
\mymath{\expectation}{\mathbb E}

\begin{theorem}[Bias] \label{thm:bias}
The bias in \refEq{unitVarianceEstimator} is $\expectation[\widehat{P_{\mathrm{Cor}}}^{(m)}] - P_{\mathrm{Cor}} = O(m^{-1})$.
\end{theorem}
\begin{proof}
Once the noise field is realized by fixing \seed, the values of all subdomains \pixelSubdomains have the fixed value \worldConstant and the area of each subdomain is \area.
Now consider the unbiased estimate of each subdomain area $\estimate\area$ that just counts the number of samples falling into area \subIndex divided by the total number of samples \sampleCount.
Let
$\areaEstimates = (\estimate{a}_1, \dots, \estimate{a}_n)$ 
and 
$\trueAreas = (a_1, \dots, a_n)$
be the vectors of all estimated and true areas, respectively.
$\estimate\correctedPixelValue$ can then be rewritten as
\begin{equation}
    \estimate\correctedPixelValue = 
    \magic(\areaEstimates)
    \text{ with }
    \magic(\mathbf{x}) = 
    \frac
    {\sum_{i=1}^\cardinality x_i \worldConstant}
    {\sqrt{\sum_{i=1}^\cardinality x_i^2}}
    .
\end{equation}
A second-order Taylor expansion of \magic about \trueAreas yields
\begin{align} \label{eq:taylor}
    &
    \expectation[\magic(\areaEstimates)]
    =
    \expectation[\magic(\trueAreas + (\areaEstimates - \trueAreas))] \\
    &=
    \expectation[\magic(\trueAreas ) + \nabla \magic(\trueAreas)^{\mathsf{T}} (\areaEstimates - \trueAreas) + \frac{1}{2} (\areaEstimates - \trueAreas)^{\mathsf{T}} H_{\magic}(\trueAreas) (\areaEstimates - \trueAreas) + R(\areaEstimates)]
    \nonumber
    \\ 
    &=
    \magic(\trueAreas) + \underbrace{\nabla \magic(\trueAreas)^{\mathsf{T}} \expectation[\areaEstimates - \trueAreas]}_{= 0} + \frac{1}{2}\sum_{i,j} \frac{\partial^2\magic}{\partial{x_i}\partial{x_j}} \Bigg|_{\trueAreas} \mathrm{Cov}(\hat{a}_i, \hat{a}_j) + \expectation[R(\areaEstimates)], \nonumber
\end{align}
where $H_{\magic}$ is Hessian and $R$ is the higher-order remainder. The expectation of the first-order term vanishes because \areaEstimates is unbiased.
To derive the covariance term, recall that each of the \sampleCount samples independently land in one of the \cardinality subdomains.
\mymath{\sampleIndictator}{X}
For sample $\omega_s$ with index $s$, the random variable $\sampleIndictator_{i,s} = \mathbb{I}_{\Omega_i}(\omega_s)$  has the following properties:
\begin{IEEEeqnarray}{rClrClrCl} 
\expectation[\sampleIndictator_{i,s}] = \expectation[\sampleIndictator_{i,s}^2] &=& a_i & \;\; \text{ and } \;\;
\expectation[\sampleIndictator_{i,s}\sampleIndictator_{j,s}] &=& 0, \label{eq:xiprops} 
\end{IEEEeqnarray}
where the left equality holds because $X_i$ is binary and the right equality holds because a sample cannot fall in two subdomains simultaneously.
Then, $\estimate\area$ can be rewritten as
\begin{equation}
    \estimate\area = \frac{1}{m}
    \sum_{s=1}^m \sampleIndictator_{i,s}.
\end{equation}
The covariance term in \refEq{taylor} becomes
\begin{align} \label{eq:cov_big}
   \mathrm{Cov}(\hat{a}_i, \hat{a}_j) &= \mathrm{Cov}\Bigg(\frac{1}{m}\sum_{s=1}^m X_{i,s} ,  \frac{1}{m}\sum_{t=1}^m X_{j,t}\Bigg) \\ 
   &= \frac{1}{m^2} \sum_{s=1}^m \Bigg(\mathrm{Cov}(X_{i,s}, X_{j,s}) + \bcancel{\sum_{\substack{t=1\\s \neq t}}^m \mathrm{Cov}(X_{i,s}, X_{j,t})}\Bigg), \nonumber 
\end{align}
where crucially the cross terms ($s \neq t$) vanish because samples are independent. 
Expanding the inner term and using \refEq{xiprops} yields
\begin{equation} \label{eq:cov_small}
    \mathrm{Cov}(X_{i,s}, X_{j,s}) = a_i \delta_{ij} - a_i a_j,
\end{equation}
where $\delta_{ij}$ is the Kronecker delta.
Substituting \refEq{cov_small} to \refEq{cov_big} yields
\begin{equation}
    \mathrm{Cov}(\hat{a}_i, \hat{a}_j) = \frac{a_i \delta_{ij} - a_i a_j}{m} = O(m^{-1}).
\end{equation}
Similarly, it can be shown that the remainder term $\expectation[R(\areaEstimates)] = O(m^{-2})$ by expanding the third moments of $\areaEstimates$, which also only include diagonal terms due to independent samples. This step is mechanical and included in \refSecSuppl{Proof}{3}.
Finally, notice that the double sum over $i,j$ in \refEq{taylor} has constant terms with respect to $m$, and that the partial derivatives of $\magic$ are also constant. Substituting both results into \refEq{taylor} completes the proof.
\end{proof}
\begin{remark*}
\refThm{consistency} and \refThm{bias} can be generalized to support prefiltering by concatenating the subdomains of two noise levels and re-normalizing their areas.
In addition, \refEq{cov_big} relies on the assumption of independent samples, which technically does not hold for \ac{QMC} sequences. Nonetheless, empirically we observe similar (or even faster) convergence, as shown next.
\end{remark*}

\mysubsection{Empirical}{Empirical}
In practice, both theorems cannot be guaranteed due to potential hash collisions of \smallHashing.
We thus conduct experiments to validate that our estimator achieves high-quality noise under reasonable settings.
We study three tasks, that correspond to different maps (flows, correlations) that we impose on the noise: a simple scroll, optical flow of a person walking a street and ray-tracing of three diffuse falling spheres.
Quality of our estimated noise depends on two variables: the number of Monte Carlo samples \sampleCount and the histogram size \binCount. \refFig{Convergence} analyzes this relation.

\myfigure{Convergence}{Convergence analysis:
Different colors mean different tasks.
In the top row, horizontal axis is sample count \sampleCount, while it is bin count \binCount in the second row.
The vertical axis is, first, \ac{MSE} of the estimator, where less is better, second, the whiteness of the spectrum \ie how similar the noise is to white noise, measured as one minus the $L_2$-distance to the flat spectrum of white noise (less is better), 
\change{third} the variance of the estimated random variable (ideally, unit variance, so, 1),
and \change{fourth and finally} the Gaussianity measured by the Shapiro-Wilk test (less is better). 
}

The first plot shows estimator \ac{MSE} for different sample counts.
We first compute a reference using a large sample count (32,768) and then produce many instances of the noise with the same seed for the random field but different seeds for the MC/QMC samples.
The error is then measured by subtracting these instances from the reference and averaging the squared distance. The \ac{MSE} is, by construction, the sum of variance and squared bias. As predicted by \refSec{Analytical}, the error decreases with more samples for \ac{MC} at an approximately linear rate. We further observe that \ac{QMC} converges faster, matching common practice in rendering.
It does so slightly slower for the ray-tracing task, in which many estimation problems are time-space averaged.
Simpler problems like zooming or optical flow converge more directly.

The second plot of \refFig{Convergence} shows the whiteness of the spectrum.
For this, many instances with different sampling seeds and different random seeds for the random field are generated and the spectrum of these is averaged and $L_2$-compared to a flat spectrum of white noise.
Ideally, all our spectra were white.
We see, that the difference to white noise is quite small, in the order of 3\,\% of the signal amplitude (on average, 1).
With increasing sample counts, this is increasing, which first looks counter-intuitive but eventually is due to increased hash collisions, which leads to these spectra deviating from white.

The third plot of \refFig{Convergence} shows the change of variance of the random variable estimated for different bin counts \binCount.
Ideally, the variance of the variable we estimate was 1, unit-variance.
We see, that for a low number of bins, this is not the case.
In particular, for a single bin, the variance is around 0.1, indicating around $10^2$ random variables are combined into a pixel resulting in too little variance.
With more bins (not samples, this is for a fixed sample count of $\sampleCount=1024$), the variance is approaching 1.
Noting the horizontal log scale, the sweet spot is around 64 bins, which is the value used in most results of this paper.

Finally, even a white spectrum and a variance of 1 does not obviously mean the resulting noise is Gaussian.
To analyze this, in the fourth and last plot of \refFig{Convergence}, we test for Gaussianity using the Shapiro-Wilk test, which returns the probability, that a sample of a certain size is indeed from a Gaussian.
A value of 1 means it certainly is from a Gaussian.
We note, that at low bin counts, our noise indeed is Gaussian, as the difference, in units of probability of being Gaussian, is small.
This property seems unaffected by the bin count.

\myfigure{Performance2}{
Performance instrumentation.
In every subplot, the vertical axis is time.
In the first row, the horizontal axis is
sample count \sampleCount,
the \change{second one} is image resolution,
and the \change{third one} is the bin count \binCount of the unique count estimator. 
All horizontal axes are in log scale.
Different colors are different tasks, while different line styles mean different underlying noise fields (with prefiltering is solid while dashed is the prefiltered grid).
The top row of plots is a different implementation (JAX), while the bottom row is the one in WebGPU.
Please see text, for an interpretation.
}

\mysubsection{Computational performance}{Performance}
The computational performance is analyzed in \refFig{Performance2}.
The JAX platform uses four A100s that produce four frames in parallel, but time is given per frame.
The WebGPU platform uses an RTX 3090 with a hardware rasterizer to render 3D scene objects.

We observe that compute time scales with sample count (first column).
This is expected, and due to parallelism, the scale is sublinear (note the log scale).
We also note that the absolute speed depends on the complexity of the map (color).
A more complex map is markedly slower, as we do not only need to compute it, but also its \texttt{jacfwd}, its Jacobian to compute pixel footprints.
This plot does not show quality, but a prefiltered grid is slower than a simple one. 
This is, as two fields are fetched and blended, a substantial part of the overall compute.
In practice, this pays off in higher quality.

Scalability in number of pixels, or image resolution (note log scale and the quadratic increase as pixel count is quadratic in resolution), is also as expected, and favourably sublinear over that range as GPU parallelism is exploited.

Finally, the third row looks at the variation of bin count \binCount and holds resolution and sample count fixed.
Here, there is only a mild relation, as more bins mostly cost memory and not compute.
Writing to a larger histogram increases memory usage, which may reduce available threads. It also increases contention vice-versa: if many samples target the same few bins, writes become serialized due to atomic pressure.
These effects seem to cancel out, leading to good scalability.

When comparing different implementations, we see that a simple WebGPU implementation has similar timing on game-optimized hardware, and also scales better in the number of samples when rendering geometry.
Unlike the JAX implementation which uses \texttt{jacfwd}, our WebGPU implementation uses hardware rasterization.
In this setting, we take advantage of built-in per-fragment derivatives to compute footprints \cite[Sec. 8.13]{OpenGLSpec} much easier and more efficient.

Finally, our method is parallel over time, so on a 96 GPU cluster, a 96 frame-animation is produced in 
$1\times\qty{2.44}{\milli\second}$, 
where other approaches would require 
$96\times\qty{2.1}{\milli\second}$ 
\cite{burgert2025gowiththeflow} or
$96\times\qty{500}{\milli\second}$ \cite{chang2024how}, respectively.

\mysubsection{Comparison}{Comparison}

Here, we compare our approach against \citet{burgert2025gowiththeflow}, the state-of-the-art method to produce similar noise \change{to find it provides similar quality and speed but random access in time.}
\begin{table}
\caption{Comparison of different methods across various metrics. Arrows indicate whether higher ($\uparrow$) or lower ($\downarrow$) values are better.}
\label{tab:Comparison}
\setlength{\tabcolsep}{4pt}
\centering
\begin{tabular}{lrrrrr}
\toprule
Method & 
LPIPS $\downarrow$ & 
SSIM $\uparrow$ & 
PSNR $\uparrow$ & 
WRP-E $\downarrow$ &
3D WRP-E $\downarrow$ \\
\midrule
Random & 0.24 & 0.62  & 23.41  & 369.73  & 790.03 \\
Fixed & 0.23  & 0.62  & 23.48  & 316.07 &  784.45 \\
GWTF & 0.24 & 0.62   & 23.44  & 276.67 &  766.99 \\
Ours & 0.24 & 0.62  & 23.67  & 274.21 &  706.25 \\
\bottomrule
\end{tabular}
\end{table}

We evaluate following the super resolution setup of \cite{burgert2025gowiththeflow} (GWTF).
We use \texttt{StableDiffusion-x4-upscaler} as the super resolution diffusion model.
We report average numbers for 50 videos from DAVIS dataset \cite{pont-tuset2017DAVIS}.
Each video is down-scaled four times and up-scaled by the diffusion model using four different types of noises: random, fixed, GWTF, and ours.
We produce our noises using an absolute 3D optical flow, \ie 2D coordinate and time for every pixel at its first appearance, estimated from the low-resolution videos with \texttt{Cotracker3} \cite{karaev24cotracker3}.
We evaluate the fidelity per frame using LPIPS, SSIM, and PSNR, local temporal consistency using the warping error (WRP-E) proposed in
\cite{laiLearningBlindVideo2018}, and global temporal consistency using a 3D warping error (3D WRP-E) which uses 3D optical flows to compute warping errors as detailed in \refSecSuppl{FlowDetails}{6}.
We conduct the experiments in five independent runs with different random seeds for generating noises in each method, and report the average.

Results are seen in \refTab{Comparison}.
In summary, we consider our performance non-inferior (\change{while being more general and random-access in time}) for the first four columns.
According to LPIPS and SSIM, neither the original method nor ours has an edge relative to trivial baselines.
For PSNR and WRP-E, ours has a small, and probably statistically significant, lead.
As measured in 3D WRP-E, our method has noticeable improvements for global temporal consistency.

\mysection{Results}{Results}

\mysubsection{Video generation}{VideoResults}
The premiere application of spatio-temporal correlated unit-variance noise is the control of diffusion-based generative video models.
\refFig{Teaser} and \refFig{Results} show still frames from such generated videos. We refer readers to the supplemental video for full sequences.

Our generation pipeline leverages the Go-with-the-Flow \cite{burgert2025gowiththeflow} text-to-video diffusion model.
The input noise, originally at $720 \times 480 \times 16$ resolution, is
spatially down-sampled by a factor of 8 to match the VAE latent space
($90 \times 60 \times 16$) using a variance-preserving uniform
down-sampling method to ensure the noise distribution remains compatible
with the diffusion model's assumptions.
Temporally, we utilize a ``direct noise'' mapping where 13 noise frames are up-sampled to 49 video frames via the VAE’s temporal decoder, bypassing additional temporal down-sampling.
For all experiments, we utilize 200 inference steps, a classifier-free guidance scale of 6.0, and a noise degradation factor of 0.2 (mixing 20\,\% random Gaussian noise) to balance temporal consistency with creative prompt adherence. 

\paragraph{Conflicts}
It is not obvious what happens if correlation and prompt for a diffusion model are at odds.
For example, is the diffusion model required to perfectly align the result to the edges of a sphere if the noise is made from a sphere seen by a rotating camera? It turns out that these two conflicting goals are often combined quite naturally.
\refFig{Conflict} shows a correlation that is at obvious contradiction to the prompt.
The correlation is the one of falling spheres while the prompt asks for an ice cube.
The resulting video shows cubes and chickens instead of spheres, yet, with the motion of the sphere.
See the supplemental video for an animated version of this result.

\paragraph{Light transport-consistency}
Our method supports advanced correlations such as the one resulting from reflection and refraction.
What effect does that have?
\refFig{LightTransport} of different prompts with different correlations induced by different light transport.
The prompt is asking for "diffuse", for "mirror" and for "glass" while the correlation is from 3D scenes which are diffuse, reflective and refractive.
When considering all combinations, the ones that match tend to provide the best results (diagonal).

\mysubsection{Stylized shading design}{ShaderDesignResult}
A common approach for producing different NPR styles is to threshold or blend noise.
The coordinate space in which the noise is defined becomes an important consideration.
2D screen-space noise is the simplest to implement but is subject to swimming or the ``shower door effect'' upon camera motion \cite{benard2011state}.
On the other hand, 3D world-space noise is stable with respect to camera motion, but the size of the noise pattern changes due to perspective projection.
To keep the noise size constant, one natural attempt is to introduce multiple levels of noises, but na\"ively switching or blending between levels would break the noise properties, causing inconsistent shading results.
This problem is solved by our approach.

In \refFig{NPR}, we demonstrate hatching stylization driven by our multi-level, variance-corrected noise.
Using regular 3D noise produces patterns that are fixed-size in world space but inconsistent on screen: too coarse for close-up parts, and too dense for distant parts.
It further causes hatching strokes to fade as variance decreases.
Our noise effectively fixes both issues and ensures constant pattern size and density everywhere regardless of object position and camera zoom.
We refer readers to our supplemental document and video for more comparisons under synchronized camera motion.

\mysubsection{Animated texture synthesis}{textureSynthesisResult}

\myfigure{TextureSynth}{Still frame of our animated texture synthesis at the left, an epipolar slice of the animation with our noise in the middle and using random noise at the right.}
We can also use our noise to initialize an optimization.
This leads to results having the desired correlation, too, without even enforcing it in the optimization loss.
An example is texture synthesis following \citet{gatys2015texture} which we perform frame-by-frame, with the only extension that the initial solution is a frame of our temporally correlated noise.
The resulting videos combine the desired textures and motion induced by the correlation. \refFig{TextureSynth} shows an animated texture with temporal correlation of two spheres falling onto a ground plane and bouncing according to a physics simulation.
The epipolar slices show the coherence and correlation in ours, while the frames in the baseline have no coherence.
Please see the supplemental video for an example of this task.

\mysubsection{Steganography}{SteganographyResult}
In steganography, a user provides a flow or a 3D scene (the ``message'') and we produce an animation.
A single frame does not reveal the message, neither to a human nor an algorithm.
Two, or more frames reveal the message to a human, who will immediately and strongly perceive the correlation, and to a machine, that can just run.
\refFig{Teaser} is an example of this, where a complex 3D structure (first row) is hidden in 2D noise with no spatial structure (third row).

\mysubsection{Inter-view correlation}{InterViewCorrelation}
As our noise is meticulously correlated also in 3D, it can be used to generate multiple views of a scene to be in correspondence.
To this end, we generate multiple noises from multiple views and generate images from them with the same prompt. 
If the views are in a lightfield rig layout, the resulting image set forms a lightfield.
The same holds for stereo pairs.
\refFig{Lightfield} shows an example of this task, where the
prompt is asking to generate "\texttt{a glass ball and Tarot cards}" inspired from the classic Stanford lightfield repository image.

\mysection{Conclusion}{Conclusion}
Our method is biased due to multiple reasons, including the nonlinear transform of estimators \cite{nimier2022unbiased,hery2018mcqmc} and the finite histogram \cite{flajolet1985probabilistic,whang1990linear}.
Developing an unbiased estimator might be a rewarding undertaking and future work.
We have targeted unit-variance as the property to retain and combine with temporal correlation, but other spatial correlation could be targeted, such as blue \cite{huang2024blue} or golden 
\cite{zhou2025golden} noise.

While the derivations in this paper can be considered involved,  the final method is simple: when using Monte Carlo to produce images that are meant to be combined with a generative model, keep track of how many random variables are accessed, and re-scale the variance by a square root and a division.

To summarize, we proposed a method to generate noise that is white-spectrum and unit variance in space, yet can follow intricate temporal correlations.
It is an extension of Monte Carlo path tracing, but can also be implemented in a rasterizer and achieves state of the art speed while being parallel over time.
We see this just as a first step to relate the statistics induced by simulation-based solutions relying on first principles, like path tracing, and the mathematical requirements of recent generative models.

\myfigure{Epipolar}{Epipolar slices (horizontal is space, vertical is time) from super-resolution result of different methods (rows).
}%
\myfigure{LightTransport}{Light transport results: the horizontal axis are prompts and the vertical axis is the correlations of the underlying noise.}

\myfigure{NPR}{Using our noise as an NPR primitive in two scenes \emph{(columns)}.
For previous methods \emph{(top)}, primitive density varies in 3D space to be consistent.
For ours \emph{(bottom)} it is kept constant and consistent.}%
\myfigure{Lightfield}{Producing correlated noise from many views can be used as input to generate lightfields (top, 4 samples of 48 shown) or stereo pairs (bottom, shown in anaglyph).}

\mycfigure{Results}{Video generation results.
The top row shows the noise, composed on top of the normal channel (not visible to the generation, only to visualize the 3D configuration).
All other rows show four frames from an animation generated from this noise using different prompts.
Note how all videos share the 3D structure, where the only guidance comes from the noise correlation.}

\mycfigure{Conflict}{Correlation of two falling spheres (seen in \refFig{Results}) in conflict with a prompt asking for ice cubes.
The generative model produces motion that complies with the one contained in our noise, but despite the correlation with falling spheres, cubes with the right motion are produced.
The supplemental video has more results for other conflicting prompts, too.}

\begin{acks}
TR would like to thanks IP for the empty slate and Thomas Leimkühler for discussion.
\end{acks}

\bibliographystyle{ACM-Reference-Format}
\bibliography{paper}

\end{document}

\title{Supplemental Material for Tabula Rasa: Monte Carlo estimation of unit-variance noise with controlled spatio-temporal correlation}

\author{Tobias Ritschel}
\affiliation{%
	\institution{University College London}
	\country{UK}
}
\affiliation{%
	\institution{Meta Reality Labs}
	\country{USA}
}
\email{t.ritschel@ucl.ac.uk}

\author{Yang Zhou}
\affiliation{%
	\institution{Meta Reality Labs}
	\country{USA}
}
\email{zhouyang@meta.com}

\author{Nick Milef}
\affiliation{%
	\institution{Meta Reality Labs}
	\country{USA}
}
\email{nickmilef@meta.com}

\author{Mikhail Dereviannykh}
\affiliation{%
	\institution{Karlsruhe Institute of Technology}
	\country{Germany}
}
\affiliation{%
	\institution{Meta Reality Labs}
	\country{USA}
}
\email{devmishok@gmail.com}

\author{Chen Liu}
\affiliation{%
	\institution{University College London}
	\country{UK}
}
\email{chen.liu.21@ucl.ac.uk}

\author{Christophe Hery}
\affiliation{%
	\institution{Meta Reality Labs}
	\country{USA}
}
\email{christ.hery@gmail.com}

\author{Carl Marshall}
\affiliation{%
	\institution{Meta Reality Labs}
	\country{USA}
}
\email{csmarshall@meta.com}

\maketitle

\mymath{\pixelValue}{P}
\mymath{\correctedPixelValue}{\pixelValue_\mathrm{Cor}}
\mymath{\otherImageCoord}{\mathbf y}
\mymath{\pixelDomain}{\Omega}
\mymath{\pixelSubdomains}{\pixelDomain_i}
\mymath{\seed}{\xi}
\mymath{\var}{\mathrm{Var}}
\mymath{\world}{f}
\mymath{\cardinality}n
\mymath{\sampleCount}m
\mymath{\subIndex}i
\mymath{\sampleIndex}j
\mymath{\area}{A_\subIndex}
\mymath{\worldConstant}{F_\subIndex}
\mymath{\estimate}{\widehat}
\mymath{\sample}{\omega_\sampleIndex}
\mymath{\hashing}{\mathcal{H}}
\mymath{\bigHashing}{\hashing_\mathrm b}
\mymath{\smallHashing}{\hashing_\mathrm s}
\mymath{\lattice}{\mathcal L}
\mymath{\noise}{\mathcal N}
\mymath{\bin}b
\mymath{\binCount}o
\mymath{\prob}p

\section*{Overview}
In this supplemental we will give in \refSec{Primer} a primer on (unit)-variance and its role in graphics, specifically Monte-Carlo estimation, provide a rasterization-based version of our approach to complement the path-tracing one from the main paper in \refSec{Pseudocode}, derive the complete equations for the bias of our estimator in \refSec{Proof}, \change{derive and discuss the residual spatial correlation due to shared random variates by pixels in \refSec{SpatialCorrelation}}, and conclude with additional result for our \ac{NPR} shader use case in \refSec{NPR}. An additional illustration of our estimation procedure is also provided in \refFig{Concept}. 

\mysection{Primer on variance and correlation}{Primer}

This section recalls the reasons of variance not being unit and the meaning of temporal correlation in our context. 

\mywfigure{Didactics}{0.5}{Variance of multiple variates.}
\paragraph{Variance}
Consider a single normally distributed, zero-mean, unit variance random variable such as flipping a coin multiple times (\refFig{Didactics}, first row).
The property of ``unit variance'' says, that the standard deviation is 1, so on average, a sample from that distribution is 1 away from the mean 0.
Now consider the average of two compared to the one of ten coin flips.
It is intuitive, and follows from the law of large numbers, that the average of these will become more similar to 0.5 and, hence, how the variance will become smaller.
The same tendency is true, if the random variables are not averaged, but a weighted sum: even if they are weighted differently, the outcome is close to that weighted combination of means and less likely to be something else.
In this case, the variance is $\sum w_i^2$, where $w_i$ are the weights.
For the special case of averaging, this means weights are equal, a mere $w_i=1/n$, and hence $\sum 1/n^2=n/n^2=1/n$, so variance goes down inversely proportional to the number of variates.

The second row of \refFig{Didactics} shows this from a graphics or image processing perspective: flipping coins is similar to producing 10 images full of noise and averaging them so as the contrast will disappear, converging to a 50\,\% gray image in the limit.
Now filtering (\refFig{Didactics}, third row) also is nothing else than a weighted sum of pixel values.
If the pixel values are random variables and we need to filter them, this will change the variance too.

This work is about adjusting the change of variance back to 1 by knowing what the variance is.
Once we know the variance of a random variable, we can divide by its square root and the variable becomes unit-normal.
Visually, this would mean to restore the contrast, and algorithmically, this is what a program such as Photoshop would do: find the contrast and rescale inverse-proportional to it.

\paragraph{Correlation}
\mywfigure{3DCorrelation}{0.39}{3D and 2D  correlation.}
Correlation in this article can be between pixels at different spatial locations and different points in time.
The aim is to introduce no unwanted spatial correlation.
For some applications, a level of spatial correlation is beneficial to prevent aliasing \cite{mitchell1988reconstruction}, while for others it is to be avoided and pixels are treated as squares \cite{smith1995pixel}. 
In our scenario, spatial correlation  will be added by downstream methods but they can only do this faithfully if they can safely assume to be absent in their input (our output).
We do however want to introduce controlled temporal correlation.

In the notion of previous papers, temporal correlation was 2D and defined by means of flow: a pixel is correlated in one frame with another pixel at a different position (black lines in \refFig{3DCorrelation}).
This has limitations at image boundaries or in the presence of occlusion, it assumes opaque diffuse objects and struggles where not all images contain all pixels that seek to be correlated and hence correlation can not be expressed as flow (\refFig{3DCorrelation}, dotted lines).
We can support this notion, but we can also define correlation resulting from directly sampling the random field in 3D (\refFig{3DCorrelation}, right) to produce projections that have correlation unaffected by occlusion or clipping.
This idea even extends to refraction and reflection of rays computed by Whitted-style ray casting.

\mycfigure{Concept}{Illustration of our estimation procedure for two pixels, 0 and 1 in a 1D-image using a ray-tracing image formation model with 8 samples in each pixel.
For each sample, a ray is traced into the 2D world (red lines) to find the intersection with 2D surfaces (blue).
Each hit coordinate is hashed to two codes: one to estimate variance (upright Greek, $\upalpha$ to $\updelta$) and one to estimate the variate (Roman, A to P).
First, the mean of all values for the first (``large'') hash codes is computed.
Next, we build a histogram of all 
the repeating (``small'') hash codes within the pixel.
This allows us to recompute the variance and adjust the estimate as per the last row.
The hashing can have collisions, as shown in the pixel where large code $J$ and $L$ both map to code $\beta$, hence overestimating the variance.}

\mysection{Rasterization pseudo-code}{Pseudocode}
\begin{algorithm}[b]
    \caption{Rasterization variant.}
    \label{alg:Shader}
    \begin{algorithmic}[1]
    \Procedure{GetNoise}{polys}
        \State pTex, dPdXTex = \Call{DrawPolys}{\textproc{Draw}, polys}
        \State vTex = \Call{FullScreenQuad}{\textproc{Normalize}, pTex, dPdXTex}
        \State\Return{vTex}
    \EndProcedure
    \Procedure{Draw}{position}
        \State \Return{position, \Call{dFdX}{position}}
    \EndProcedure
    \Procedure{Normalize}{fragCoord, pTex, dPosdXTex}
        \State histoSize = 64
        \State weights = float[2,histoSize]
        \State value = 0
        \For{offset in (size, size)}
            \State p = \Call{TexelFetchOffset}{pTex, fragCoord, offset}
            \State foot = \Call{TexelFetchOffset}{dPdXTex, fragCoord, offset}   
            \State level = \Call{log2}{foot}
            \State dLevel = \Call{floor}{level}
            \State id0 = \Call{hash}{dLevel + 0, p}
            \State id1 = \Call{hash}{dLevel + 1, p}
            \State value0 = \Call{Noise}{id, dLevel + 0}
            \State value1 = \Call{Noise}{id, dLevel + 1}
            \State value += \Call{LERP}{value0, value1, weight}
            \State weight = \Call{Fract}{level}
            \State weights[0, id0 \% histoSize] += weight 
            \State weights[1, id1 \% histoSize] += 1 - weight 
        \EndFor
        \State weights /= \Call{Sqr}{size}
        \State value /= \Call{Sqr}{size}
        \State sqrWeights = \Call{Sqr}{weights}
        \State variance = \Call{Sum}{sqrWeights}
        \State \Return{value / \Call{Sqrt}{variance}}
        \EndProcedure
    \end{algorithmic}
\end{algorithm}
We provide a variant of our method implemented in the rasterization shader as \refAlg{Shader} to complement the raytracing-based one from the main paper (Main~Alg.~1).

\mymath{\magic}{\psi}
\mymath{\areaEstimates}{\estimate{\mathbf a}}
\mymath{\trueAreas}{\mathbf a}
\mymath{\expectation}{\mathbb E}
\mymath{\taylorpower}{p}

\mysection{Additional details on the proof of Main~Thm.~2}{Proof}
Here we provide the full derivation to bound the remainder term in Main~Eq.~10.
\begin{lemma} \label{thm:remainder}
$\expectation[R(\areaEstimates)] = O(m^{-2})$.
\end{lemma}
\begin{proof}
Recall that $R(\areaEstimates)$ is the remainder of the second-order Taylor expansion of $\magic(\mathbf{x})$ about $\trueAreas$. By the mean value theorem, the Lagrange form of $R(\areaEstimates)$ is
\begin{equation}
R(\areaEstimates) = \frac{1}{6}\sum_{i,j,k} \frac{\partial^3 \magic}{\partial x_i \partial x_j \partial x_k} \Bigg|_{\mathbf{c}} (\hat{a}_i - a_i)(\hat{a}_j - a_j)(\hat{a}_k - a_k),
\end{equation}
where $\mathbf{c}$ is some point on the segment between $\trueAreas$ and $\areaEstimates$: $\mathbf{c} = \trueAreas + t_c(\areaEstimates - \trueAreas), t_c \in (0, 1)$. Taking the expectation on both sides gives
\begin{equation}
    \expectation[R(\areaEstimates)] =  \frac{1}{6} \expectation \Bigg[\sum_{i,j,k} \frac{\partial^3 \magic}{\partial x_i \partial x_j \partial x_k} \Bigg|_{\mathbf{c}} (\hat{a}_i - a_i)(\hat{a}_j - a_j)(\hat{a}_k - a_k)\Bigg].
\end{equation}
Unlike the second-order partial derivatives in Main~Eq.~10, the third-order partial derivatives here depend on $\areaEstimates$ through $\mathbf{c}$. 
However, we can factor them out as they are bounded due to the smoothness of $\phi(\mathbf{x})$ ($C^{\infty}$ on $(0, \infty)^n$). Formally, let $M_{ijk}$ be the supremum of the absolute values of each partial derivative over an epsilon-open ball around $\mathbf{a}$ that covers $\mathbf{c}$, and let $M$ be the maximum of all $M_{ijk}$:
\begin{equation}
M_{ijk} = \sup_{\mathbf{c} \in \mathcal{B}_{\epsilon}(\trueAreas)} \Bigg| \frac{\partial^3 \magic}{\partial x_i \partial x_j \partial x_k} \Bigg|_{\mathbf{c}} \Bigg| < \infty \; \text{and} \; M = \max_{i,j,k} M_{ijk} < \infty.
\end{equation}
Then, $\expectation[R(\areaEstimates)]$ is bounded by
\begin{equation} \label{eq:third_moment}
\expectation[R(\areaEstimates)] \leq \frac{M}{6} \sum_{i,j,k} \expectation[(\hat{a}_i - a_i)(\hat{a}_j - a_j)(\hat{a}_k - a_k)].
\end{equation}
Recall the definition of $X_{i,s}$ in the main paper. We now further define random variable $Z_{i,s} = X_{i,s} - a_i$ and write
\begin{equation}
    \hat{a}_i - a_i = \frac{1}{m} \sum_{i=1}^m Z_{i,s}.
\end{equation}
The third moment term inside \refEq{third_moment} then becomes
\begin{equation} \label{eq:triple_sum}
    \expectation[(\hat{a}_i - a_i)(\hat{a}_j - a_j)(\hat{a}_k - a_k)] = \frac{1}{m^3} \sum_{s=1}^m \sum_{t=1}^m \sum_{u=1}^m \expectation[Z_{i,s} Z_{j,t} Z_{k,u}].
\end{equation}
Because samples $\omega_1, \dots, \omega_m$ are independent and $\expectation[Z_{i,s}] = 0$, only the diagonal terms ($s = t = u$) in \refEq{triple_sum} are nonzero:
\begin{equation} \label{eq:triple_sum_diagonal}
    \expectation[(\hat{a}_i - a_i)(\hat{a}_j - a_j)(\hat{a}_k - a_k)] = \frac{1}{m^3} \sum_{s=1}^m \expectation[Z_{i,s} Z_{j,s} Z_{k,s}].
\end{equation}
To derive the expression of the inner term, we generalize and utilize the properties of $X_{i,s}$ from Main~Eq.~11. Given a list of sub-pixel domain indices $i_1, \dots, i_r$ and $q_1, \dots, q_r \geq 1$,
\begin{equation} \label{eq:ext_prop}
\expectation[X_{i_1,s}^{q_1}, \dots X_{i_r,s}^{q_r}] =
\begin{cases} 
a_{i_1} & \text{if } i_1 = i_2 = \dots = i_r, \\
0 & \text{otherwise.}
\end{cases}
\end{equation}
\refEq{ext_prop} holds because each $X_{i,s}$ is binary and each sample $\omega_s$ cannot fall in two subdomains simultaneously.
In particular, $\expectation[X_{i,s}^{q}] = a_i$ for any $q \geq 1$.
The inner term of \refEq{triple_sum_diagonal} can now be expressed in cases:
\begin{enumerate}
    \item $i=j=k$:
\begin{align} \label{eq:case1}
    \expectation[(Z_{i,s})^3] &= \expectation[(X_{i,s} - a_i)^3] \nonumber \\
    &= \expectation[ X_{i,s}^3 - 3X_{i,s}^2 a_i + 3 X_{i,s} a_i^2  - a_i^3 ] \nonumber \\
    &= 2a_i^3 -3a_i^2+a_i.
\end{align}
    \item Only two indices are equal. Without loss of generality, let $i=j \neq k$:
\begin{align} \label{eq:case2}
    \expectation[(Z_{i,s})^2 Z_{k,s}] &= \expectation[(X_{i,s} - a_i)^2 (X_{k,s} - a_k)] \nonumber \\
    &= 2 a_i^2 a_k - a_i a_k.
\end{align}    
    \item $i \neq j$ and $j \neq k$:
\begin{align} \label{eq:case3}
    \expectation[Z_{i,s} Z_{j,s} Z_{k,s}] &= \expectation[(X_{i,s} - a_i) (X_{j,s} - a_j) (X_{k,s} - a_k)] \nonumber \\
    &= 2a_i a_j a_k.
\end{align}
\end{enumerate}
Notice that in all three cases, the expression is constant with respect to $m$. Therefore, substituting \refEq{case1} -- \refEq{case3} into the right hand side of \refEq{triple_sum_diagonal} gives a sum of $m$ constant terms divided by $m^3$, which is $O(m^{-2})$. Finally, the proof is completed by substituting \refEq{triple_sum_diagonal} into \refEq{third_moment}.
\end{proof}

\mysection{Spatial correlation}{SpatialCorrelation}
\myfigure{SpatialCorrelation}{Residual spatial correlation: any method that resamples a pixel grid (two pixels U and V shown as big squares) to a non-aligned texel grid (colored areas) is subject to some form of correlation.
Consider the different values at each texel, visualized as colors.
As the texel and pixel grid do not align, there are always some texels that are shared by two pixels, in this case $x_2$, $x_7$, and $x_{10}$.
The resulting pixel value is a combination of colors, which all are, without loss of generality, distinct, besides the shared ones.}

There is one source of minor spatial correlation in our approach and probably also in other approaches that make use of resampling from a high, but finite-resolution grid to a non-aligned pixel grid such as \citet{burgert2025gowiththeflow}.

As shown in \refFig{SpatialCorrelation}, a pixel grid and texels are, without loss of generality, never aligned.
Texels, under general warps, and in particular when projected from 3D as in our method, are not even a regular grid.
When a pixel samples multiple texels as a linear combination given by reconstruction filter weights, it will always, however high the resolution is, touch one or more texels that are also touched by another pixel.
In our illustration this is texel 2, 7 and 10.
Fortunately, this correlation can be quantified and made arbitrarily low with arbitrary texel resolution as we show next:

\subsection{Correlation pixels due to a shared texel}

Let $\mathbf{X} = [X_1, X_2, \dots, X_n]^T$ be the random vector of all texels in the entire random field. 
We define two pixels $U$ and $V$ as distinct linear combinations of $\mathbf{X}$, using the weight vectors $\mathbf{a} = [a_1, a_2, \dots, a_n]^T$ and $\mathbf{b} = [b_1, b_2, \dots, b_n]^T$:
\[
U = \mathbf{a}^T\mathbf{X} = \sum_{i=1}^n a_i X_i
\qquad
\text{and}
\qquad
V = \mathbf{b}^T\mathbf{X} = \sum_{i=1}^n b_i X_i,
\]
\change{where $\sum_i a_i = \sum_i b_i = 1$,} 
as a pixel is a partition of unity of all texels below it.
Most elements in $\mathbf{a}$ and $\mathbf{b}$ are zero as a pixel does not cover most texels.
If it is non-zero, it is almost always non-zero in only either $\mathbf{a}$ or $\mathbf{b}$ as one texel only falls on one pixel.
The exceptions are those texels, that affect both $U$ and $V$. These are the ones that induce a subtle linear Pearson correlation between the pixel variables.

\change{
While the individual elements of $\mathbf{X}$ have unit variance, $U$ and $V$ being their linear combinations need not, and correcting for this is precisely the paper's objective. Let their variance-corrected counterparts be
\begin{equation*}
U' = \frac{U}{\sqrt{\mathrm{Var}(U)}}
\qquad
\text{and}
\qquad
V' = \frac{V}{\sqrt{\mathrm{Var}(V)}}.
\end{equation*}
Our method is about jointly estimating $U$ and $\mathrm{Var}(U)$ (resp. $V$ and $\mathrm{Var}(V)$) via Monte Carlo, but for the sake of this discussion, we assume perfect estimates and no Monte Carlo noise. The Pearson correlation between $U'$ and $V'$ can then be quantified as
}
\change{
\[
R=
\text{Corr}(U', V') = \frac{\text{Cov}(U, V)}{\sqrt{\text{Var}(U)\text{Var}(V)}}.
\]
}
To compute this expression, we review three properties carried from the main paper:
\begin{enumerate}
    \item \textbf{The texel noise has unit variance:} $\text{Var}(X_i) = 1$ for all $i$ as every noise texel simply is unit-random value.
    \item \textbf{The texel noise is uncorrelated:} $\text{Cov}(X_i, X_j) = 0$ for all $i \neq j$ as texels simply are independent of each other.
    \item \change{\textbf{Variance of each pixel:} $\text{Var}(U) = \sum_i a_i^2 = ||\mathbf{a}||^2$ and $\text{Var}(U) = \sum_i b_i^2 = ||\mathbf{b}||^2$. This is carried from Main~Eq.~6.}
\end{enumerate}
From \change{Properties} 1 and 2, the covariance matrix of $\mathbf{X}$ simplifies to the identity matrix ($\mathsf{\Sigma} = \mathsf{I}$). 
By substituting $\mathsf{\Sigma} = \mathsf{I}$ into the joint covariance formula, the numerator simplifies to:
\[
\text{Cov}(U, V) = \mathbf{a}^T \mathsf{\Sigma} \mathbf{b} = \mathbf{a}^T \mathsf{I} \mathbf{b} = \mathbf{a}^T \mathbf{b}.
\]

\change{Together with Property 3, the correlation becomes:
\begin{equation} \label{eq:spatial_correlation}
R = \frac{\mathbf{a}^T \mathbf{b}}{||\mathbf{a}|| \, ||\mathbf{b}||}.
\end{equation}
This means, that the correlation between two of our pixels is the \emph{cosine similarity of their normalized weight vectors}.
}

\subsection{Example}
Let us consider two examples to get an idea about the actual correlation values.

Consider two pixels observing an opaque diffuse world.
In that case both $\mathbf{a}$ and $\mathbf{b}$ are mostly zeros, and they are rarely non-zero at the same index. The correlation comes only from those rare coinciding non-zero components.
Hence, only these coinciding components matter as otherwise one or both factors in a product is zero.

To see how large they are, consider a 1D case in \refFig{Correlation1D}, where for example 10 texels fall onto one pixel.
Only one texel $k$ is shared by both pixels in 1D. 
\change{Let $w$ be the length of this texel that straddles $U$, which can vary between $0$ and $0.1$. 
The spatial correlation can be written precisely according to \refEq{spatial_correlation},
with $||\mathbf{a}|| = ||\mathbf{b}||$ because the texel at each pixel's other boundary is split in the same way into $w$ and $0.1 - w$:
\begin{equation*}
    R = \frac{w(0.1-w)}{(0.1 - w)^2 +  9 \cdot 0.1^2 + w^2}, \quad w \in (0, 0.1).
\end{equation*}
which peaks at approximately $0.026$ when $w = 0.05$, which is once texel $k$ covers both pixels in equal share.
\myfigure{Correlation1D}{An example of 1D spatial correlation.}
}

\myfigure{CorrelationBoundary}{A texel grid of doubling resolution resp.\ magnification (columns), falling onto a single pixel (yellow square).
The first row shows the configuration where texels that would induce correlation with neighbouring pixels are marked orange.
The second row shows the respective area of unique vs.\ correlated texels.
We see that with typical resolution the correlated pixel areas vanishes and \change{so does the correlation.}
}

In the general 2D/3D case for general mappings from texels to pixels, as in \refFig{SpatialCorrelation}, it is difficult to quantify how texels exactly project onto both of two pixels.
The tendency however is clear: when the magnification grows, the number of unique texels inside grows much quicker than the number of correlated texels on the edge.
\change{As shown in \refFig{CorrelationBoundary}, with $n \times n$ texels covered per pixel, the shared edge strip has approximately $n$ texels. We can roughly estimate the spatial correlation by assuming each texel weights $1/n^2$, and each shared texel is cut exactly in halves. Again using \refEq{spatial_correlation},
\begin{equation*}
    R \approx \frac{n \, (\frac{0.5}{n^2})^2}{ n^2 \, (\frac{1}{n^2})^2} = \frac{1}{4n} = O(n^{-1}).
\end{equation*}
In other words, the residual correlation is small and vanishes linearly as the texel resolution grows.
}

\subsection{Discussion}
This effect does not go away by other reconstruction filter.
It rather is becoming worse, as in this condition, simply, a single texel maps to even more pixels.
Notably, the inverse is not true: it is not sufficient to use non-overlapping reconstruction filters to guarantee no correlation.
It just appears to minimize these.
Our method, as any path tracer, is able to use any reconstruction filter, and while we have only experimented with box filters, exploring the resulting reconstruction-decorrelation continuum seems relevant to address in future work.

Note, that there are other forms of spatial correlation, which are \emph{not} due to reconstruction filtering, which might be desirable.
Those are correlations induced by light transport, \eg a pixel with a world space position that is reflecting another world space position will have, and should have, the same random variable and hence be correlated.

Depending on which integrator is used, additional correlation might be introduced, \eg by QMC, which again could be defeated by scrambling, if necessary.

We thank the anonymous reviewers for pointing out this observation of unwanted spatial correlation.
While formally correct, as explained above, it is likely so small that we did not observe any resulting issues.

\mysection{Additional results of non-photorealistic rendering}{NPR}
In \autoref{fig:new_npr_results} and the supplemental video, we provide additional \ac{NPR} results with two styles. The first style mimics point stippling by applying binary thresholding to the noise based on tone intensity. Color is achieved by independent thresholding of each channel of the CMYK color space with different noise seeds. The second style mimics hatching by blending anisotropically stretched noise layers\footnote{Technically, our noise is no longer unit-variance Gaussian after the stretching and blending. Nonetheless, the qualitative advantage of our noise remains evident.}. We compare the shading results using our multi-level, variance-corrected noise versus regular 3D noise, both defined in world space. Our approach ensures consistent pattern size and density everywhere regardless of object position and camera zoom. Orthogonal to our implementation, it is possible to further support object animation using UV-space noise~\citep{praun2001hatch} and to improve sample uniformity using blue noise~\citep{wei2010bluenoise}. We leave these potential extensions as future work.

\begin{figure*}
    \centering
    
    \setlength{\tabcolsep}{1pt}
    \renewcommand{\arraystretch}{0.5}
    \begin{tabular}{ccccc}
        & \small{\textsf{Baseline}} & \small{\textsf{Ours}} & \small{\textsf{Baseline}} & \small{\textsf{Ours}}\\
        \\
        \cline{2-5}
        \\
        & \multicolumn{2}{c}{\small{\textsc{Still Life}}} & \multicolumn{2}{c}{\small{\textsc{Winter}}}\\
        
        \multirow{2}{*}[1em]{\rotatebox{90}{\small{\textsf{Stippling}}}} &
        \includegraphics[width=0.24\linewidth]{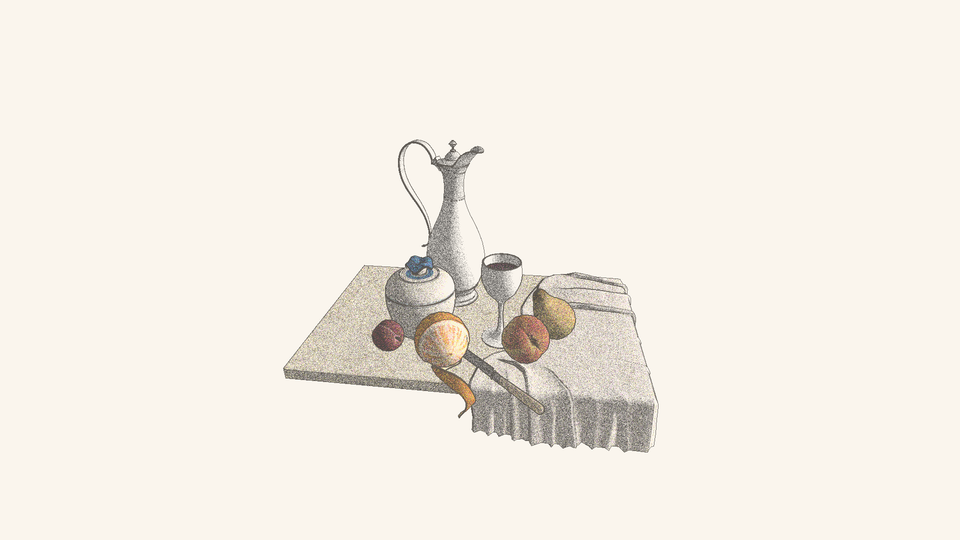} &        \includegraphics[width=0.24\linewidth]{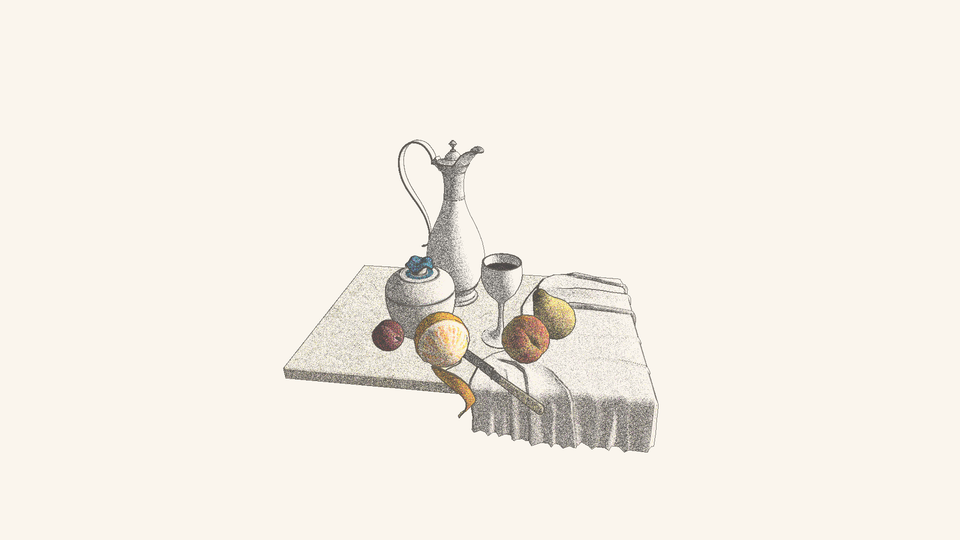} &        \includegraphics[width=0.24\linewidth]{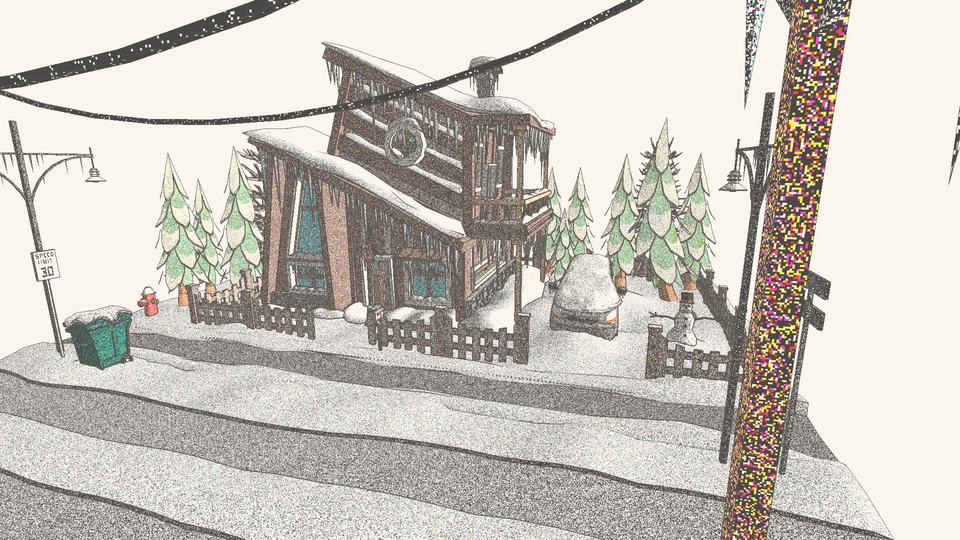} &        \includegraphics[width=0.24\linewidth]{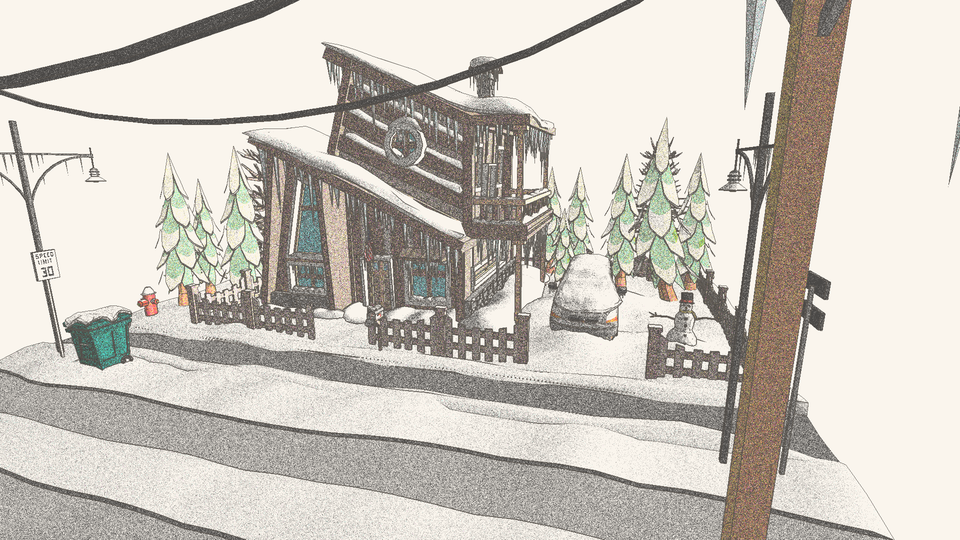} \\
        &
        \includegraphics[width=0.24\linewidth]{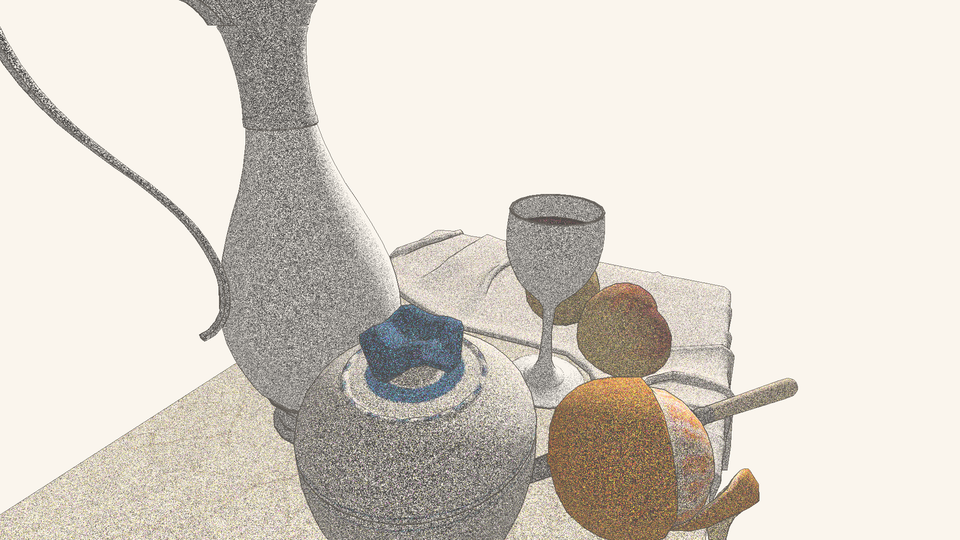} &        \includegraphics[width=0.24\linewidth]{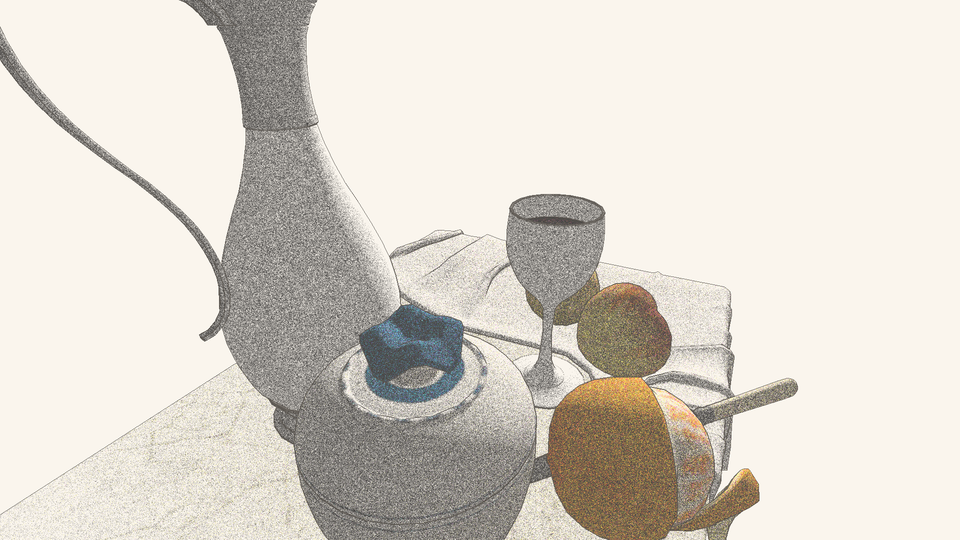} &        \includegraphics[width=0.24\linewidth]{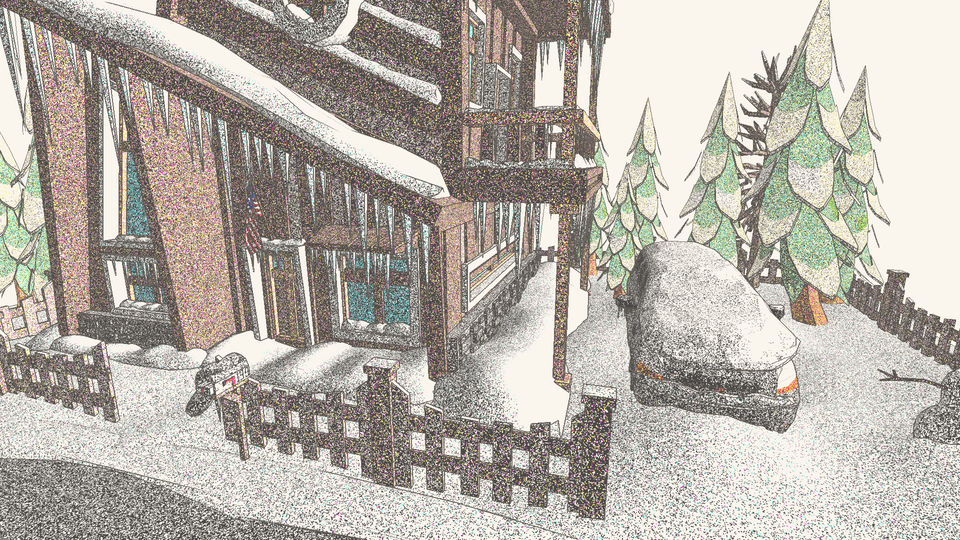} &        \includegraphics[width=0.24\linewidth]{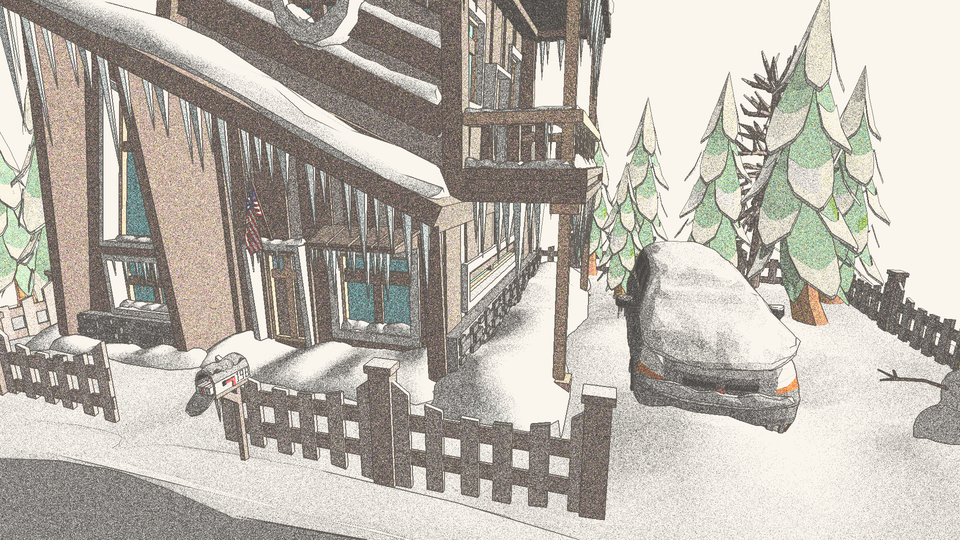} \\
        
        \multirow{2}{*}[1em]{\rotatebox{90}{\small{\textsf{Hatching}}}} & 
        \includegraphics[width=0.24\linewidth]{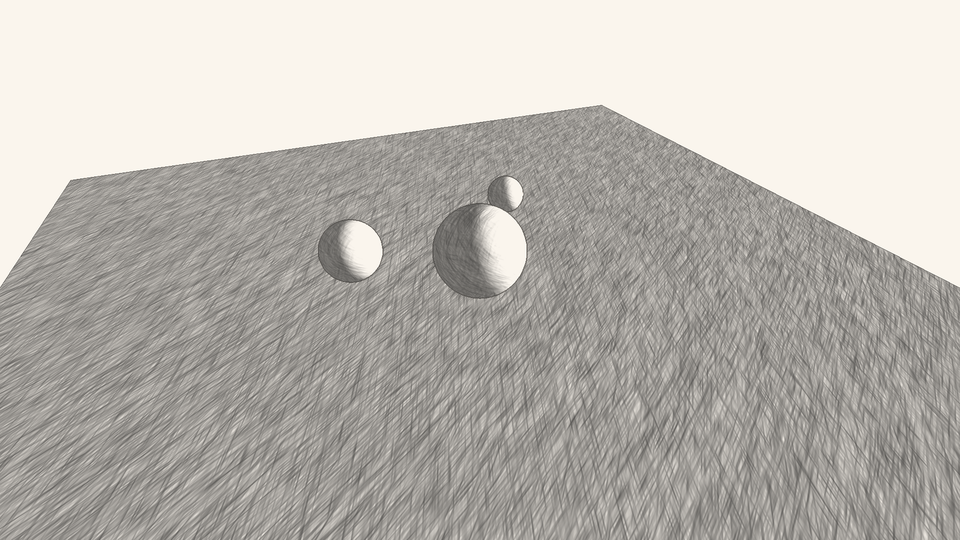} &        \includegraphics[width=0.24\linewidth]{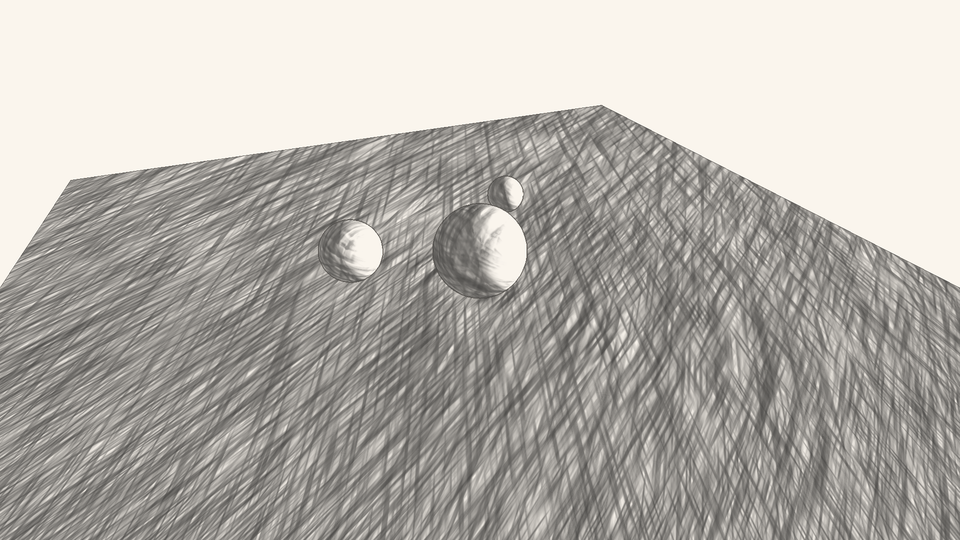} &        \includegraphics[width=0.24\linewidth]{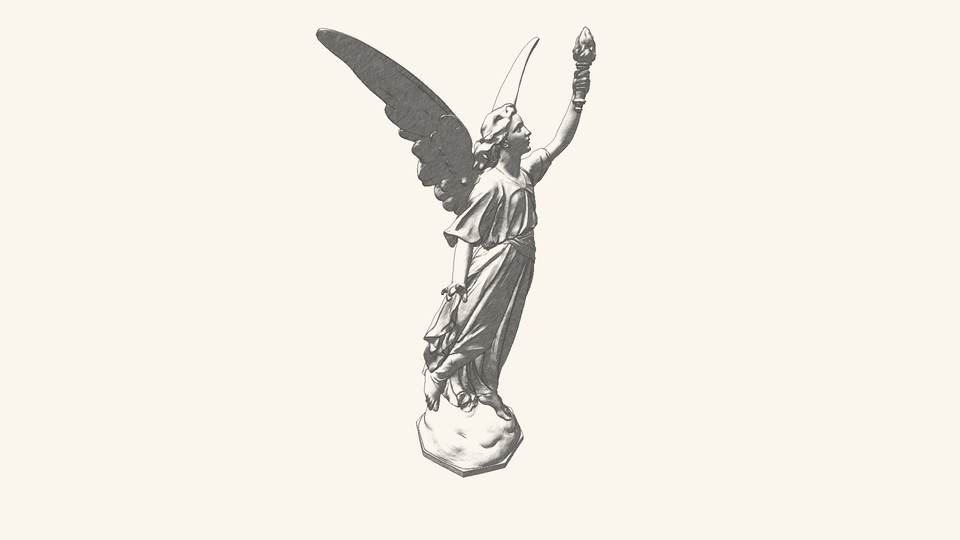} &        \includegraphics[width=0.24\linewidth]{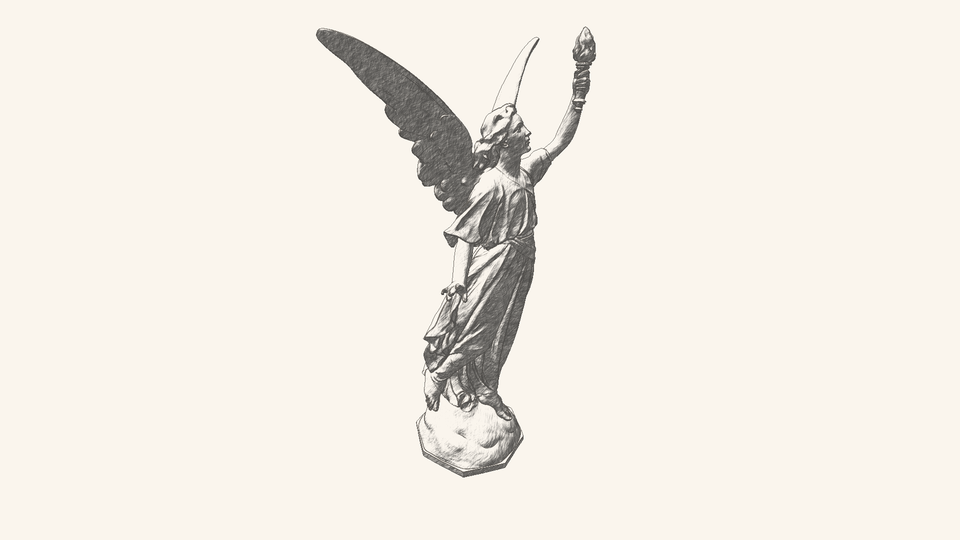} \\
        & 
        \includegraphics[width=0.24\linewidth]{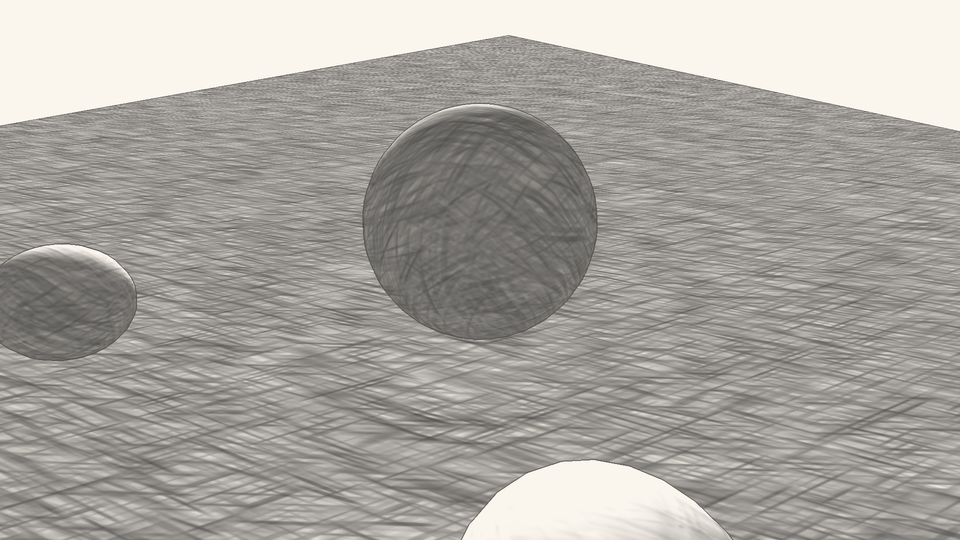} &        \includegraphics[width=0.24\linewidth]{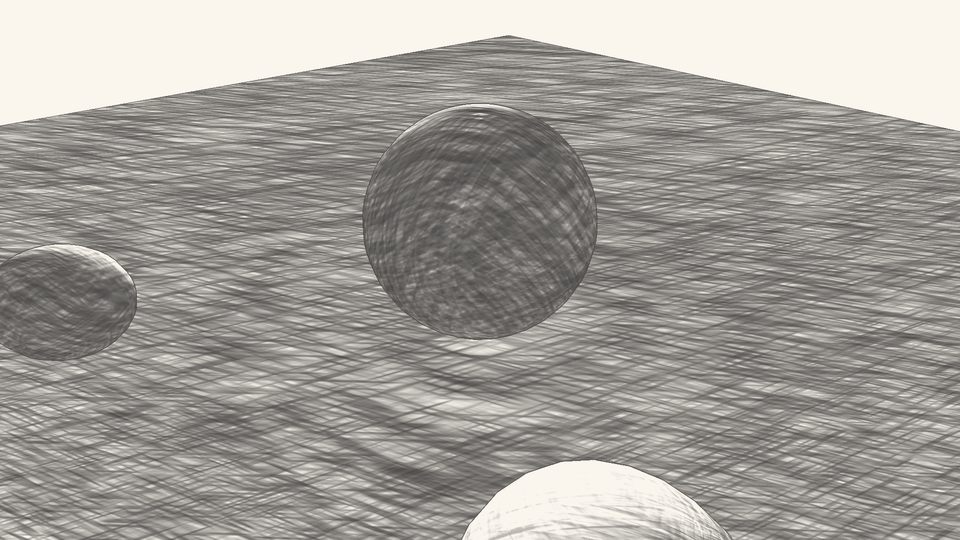} &        \includegraphics[width=0.24\linewidth]{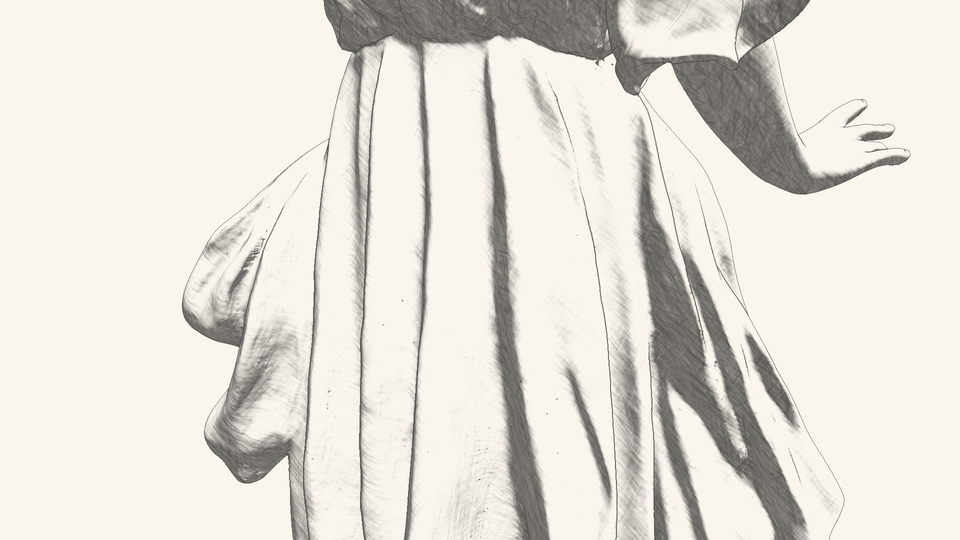} &        \includegraphics[width=0.24\linewidth]{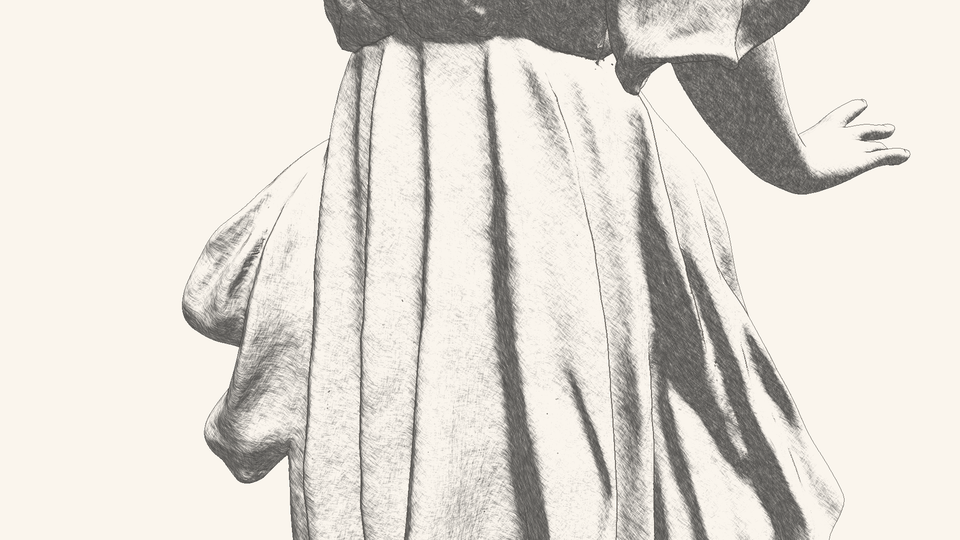} \\
        & \multicolumn{2}{c}{\small{\textsc{Playground}}} & \multicolumn{2}{c}{\small{\textsc{Lucy}}}
        
    \end{tabular}
    
    \caption{Comparisons of non-photorealistic stippling and hatching renders based on world-space noise (``Baseline'') versus our multi-level, variance-corrected noise (``Ours''). For each scene, two views are shown to highlight the different noise behaviors at different zoom levels. Notice how ours maintains a consistent scale of patterns. Readers are encouraged to zoom in for a closer inspection.}
    \label{fig:new_npr_results}
\end{figure*}

\mysection{2D vs 3D flow}{FlowDetails}
There is a subtle difference between 2D flow and 3D flow which fundamentally affects how our method and previous methods can take them into account in video restoration.
Flow can be either absolute or relative, forward or backward, and -- for the simplicity of naming -- 2D or 3D.

None of the following are assumptions of our method at all.
It is just that our method can take in more advanced input than previous work and it is different in a subtle way.

\paragraph{Absolute vs relative}
Absolute can be converted to relative and back by simple additions.
All previous as well as our method can work with them equally well as they can be converted up to numerics.

\paragraph{Forward vs backward}
Forward and backward flow can also be consumed by all methods, including ours.
It is harder to convert them, as it ideally requires a reverse mapping of a function that is supposed to be defined and evaluated smoothly in space but not practical in reality due to finite grids.

\paragraph{2D vs 3D}
The substantial difference is what the flow values actually mean.
In normal 2D flow, we have an image coordinate we can fetch from and do so with all sophistication proposed in previous work.
Our method as well can take in 2D flow.

3D flow however, is different.
Let us assume forward, absolute 3D flow as produced by \citet{karaev24cotracker3}: Now, every pixel has a unique identifier.
This is ``3D'' in so far, as it will not be affected by occlusion and hence gives every point a unique and stable identifier.
This identifier for our work does not even need to be the coordinate, it does not need to be metric even, it just needs to be unique, because this is all that our method needs to recognize correlation in space and time via hashing.
It turns out, that the identifier mostly is the image space coordinate, and the time code of first appearance.
Like this, not every 3D point always gets the same unique code, but every 3D point the flow algorithm managed to recognize does.
Taking in this flow is not possible with previous methods: if a point was occluded (or otherwise dropped) in the frame before, it is lost forever.
Previous methods simply have no way of taking in time (or any other disambiguating information), all they consume is coordinates.

\paragraph{Implementation}
In practice, we need absolute backward 3D flow, but modern flow estimation models come as point tracking.
We are not aware of a method to perform 3D flow estimation so we devised our own.
We provide a supplemental website viewer to visualize the optical flows by this new method and subsequently the noises generated based on the flows.
The new method works as follows:
for every frame the 3D flow algorithm provides a sparse set of points that in this frame emerged and for each element of that set the varying-length dense trajectory as a list of future 2d positions.
We parse the sequence start to end and for every frame we iterate over every element of the set of new emergent features.
This element is a trajectory and we splat the pair of initial 2D image coordinate and time code to the future positions in future frames.
Note, that this does not create holes, as something that is never covered by a path in the future, will spawn its own trajectory (\ie it is a disocclusion).
We splat the $x,y,t$ coordinates additively and keep a count of how many paths splatted to a pixel, to divide the coordinate sum by this count in the end.
This works well except at edges of disocclusion, because in their vicinity, both foreground and background trajectories might contribute, leading to an unwanted averaging of foreground and background motion which is correct for neither.
Instead, we would be more interested to break the tie by consistently choosing a mode of the distribution of contributors.
This can be achieved by replacing the additive blending followed by normalization by a full-fledged mode-seek: instead of adding up pixels, we keep a set of all coordinates affecting a pixel.
Once splatting is finished, this set is clustered, and the dominant mode is used as the coordinate.
Like this, pixels on the edges of moving objects either get mapped to the foreground or background, but not to an undesirable mix.

\paragraph{3D Warping Error}
We measure 3D warping errors by inheriting the main idea of warping error \cite{laiLearningBlindVideo2018} but replacing its flow estimated between two consecutive frames with the 3D optical flow.
It is a reference-based metric where the estimation of 3D optical flow requires the original reference video, and then applied to the video to be tested to compute the MSE between every pixel and its birth pixel.
Specifically, we compute as below:
\begin{equation}
    \text{MSE}(V(x) - V(F(x))),
\end{equation}
where $V$ is the video, $x$ is a valid pixel of $V$, $F$ is the 3D optical flow which records for each pixel its birth coordinate of space and time.
This measures the differences of each pixel from what it was at birth to capture the global temporal consistency, \ie space-time memory, for our needs to demonstrate the benefits of using 3D optical flow to generate our noises.
Practically, we conduct the same protocol of computing warping errors that measures MSE in RGB values of $[0, 255]$, reduces across channels, and averages for all pixels of one video.

\paragraph{Optical Flow Support}
We further summarize in \refTab{optical_flow_modality} the different optical flow modalities that our method and the baseline use to compute metrics, as well as those they can support in principle.

\begin{table}[htb]
    \centering
    \caption{Optical Flow Modality comparison.}
    \label{tab:optical_flow_modality}
    \begin{tabular}{@{}lllcc@{}}
        \toprule
        \multicolumn{3}{c}{{Modality of Optical Flows}} & {Used by} & {Supported by} \\ \cmidrule(r){1-3}
        {Type} & {Direction} & {Dim.} & & \\
        \midrule
        \multirow{2}{*}{Absolute} & Forward & 2D & ---  & Ours \\
                                  & Backward & 2D & ---  & Ours \\
        \addlinespace
        \multirow{2}{*}{Relative} & Forward & 2D & GWTF & GWTF \\
                                  & Backward & 2D & GWTF & GWTF \\
        \addlinespace
        \multirow{2}{*}{Absolute} & Forward & 3D & ---  & Ours \\
                                  & Backward & 3D & Ours & Ours \\
        \bottomrule
    \end{tabular}
\end{table}

\bibliographystyle{ACM-Reference-Format}
\bibliography{paper}